\documentclass[11pt,letterpaper]{article}
\usepackage[lmargin=1.0in,rmargin=1.0in,bottom=1.0in,top=1.0in,twoside=False]{geometry}

\usepackage[utf8]{inputenc}
\usepackage[T1]{fontenc}
\usepackage{lmodern}

\usepackage{fullpage,amssymb,amsmath}
\usepackage{graphicx}
\usepackage{enumerate}
\usepackage{paralist}
\usepackage{tikz}
\usepackage{authblk}
\usetikzlibrary{shapes,patterns}

\usepackage{xcolor}
\usepackage{mathtools}
\usepackage{microtype}
\usepackage{amsfonts}
\usepackage{comment}
\usepackage[english]{babel}
\usepackage{mathrsfs}
\usepackage[ruled,vlined]{algorithm2e}
\usepackage{blindtext}
\usepackage{thmtools}

\usepackage{trimspaces}
\usepackage{nccfoots}
\usepackage{setspace}
\usepackage{inconsolata}
\usepackage{libertine}
\usepackage[absolute]{textpos}

\usepackage{enumitem}
\usepackage{todonotes}
\usepackage{longtable}

\definecolor{blue}{rgb}{0.1,0.2,0.5}
\definecolor{brown}{rgb}{0.6,0.6,0.2}
\usepackage[ocgcolorlinks, linkcolor={blue}, citecolor={brown}]{hyperref}
\usepackage[amsmath,amsthm,thmmarks,hyperref]{ntheorem}
\usepackage{enumerate}
\usepackage{latexsym}

\usepackage{cleveref}

\crefformat{page}{#2page~#1#3}%
\Crefformat{page}{#2Page~#1#3}%
\crefformat{equation}{#2(#1)#3}%
\Crefformat{equation}{#2(#1)#3}%
\crefformat{figure}{#2Figure~#1#3}%
\Crefformat{figure}{#2Figure~#1#3}%
\crefformat{section}{#2Section~#1#3}
\Crefformat{section}{#2Section~#1#3}
\crefformat{chapter}{#2Chapter~#1#3}
\Crefformat{chapter}{#2Chapter~#1#3}
\crefformat{chapter*}{#2Chapter~#1#3}
\Crefformat{chapter*}{#2Chapter~#1#3}
\crefformat{part}{#2Part~#1#3}
\Crefformat{part}{#2Part~#1#3}
\crefformat{enumi}{#2(#1)#3}
\Crefformat{enumi}{#2(#1)#3}
\crefformat{cthmin}{#2Theorem~#1#3}
\Crefformat{cthmin}{#2Theorem~#1#3}
\newtheorem{cthmin}{Theorem}%
\newenvironment{cthm}[1]
  {\renewcommand\thecthmin{#1}\cthmin}
  {\endcthmin}

\usepackage[mathlines]{lineno}

\newcommand*\patchAmsMathEnvironmentForLineno[1]{%
  \expandafter\let\csname old#1\expandafter\endcsname\csname #1\endcsname
  \expandafter\let\csname oldend#1\expandafter\endcsname\csname end#1\endcsname
  \renewenvironment{#1}%
     {\linenomath\csname old#1\endcsname}%
     {\csname oldend#1\endcsname\endlinenomath}}%
\newcommand*\patchBothAmsMathEnvironmentsForLineno[1]{%
  \patchAmsMathEnvironmentForLineno{#1}%
  \patchAmsMathEnvironmentForLineno{#1*}}%
\AtBeginDocument{%
  \patchBothAmsMathEnvironmentsForLineno{equation}%
  \patchBothAmsMathEnvironmentsForLineno{align}%
  \patchBothAmsMathEnvironmentsForLineno{flalign}%
  \patchBothAmsMathEnvironmentsForLineno{alignat}%
  \patchBothAmsMathEnvironmentsForLineno{gather}%
  \patchAmsMathEnvironmentForLineno{split}
  \patchBothAmsMathEnvironmentsForLineno{multline}}

\theoremnumbering{arabic}
\theoremstyle{plain}
\theoremsymbol{}
\theorembodyfont{\itshape}
\theoremheaderfont{\normalfont\bfseries}
\theoremseparator{.}

\newtheorem{theorem}{Theorem}
\crefformat{theorem}{#2Theorem~#1#3}
\Crefformat{theorem}{#2Theorem~#1#3}

\newcommand{\newtheoremwithcrefformat}[2]{%
  \newtheorem{#1}[theorem]{#2}%
  \crefformat{#1}{##2\MakeUppercase#1~##1##3}%
  \Crefformat{#1}{##2\MakeUppercase#1~##1##3}%
}
\newcommand{\newseptheoremwithcrefformat}[2]{%
  \newtheorem{#1}{#2}%
  \crefformat{#1}{##2\MakeUppercase#1~##1##3}%
  \Crefformat{#1}{##2\MakeUppercase#1~##1##3}%
}

\newtheoremwithcrefformat{lemma}{Lemma}

\newtheoremwithcrefformat{proposition}{Proposition}
\newtheoremwithcrefformat{observation}{Observation}
\newseptheoremwithcrefformat{conjecture}{Conjecture}
\newtheoremwithcrefformat{corollary}{Corollary}
\newtheorem*{claim*}{Claim}

\theoremstyle{definition}
\newtheorem*{example*}{Example}
\theorembodyfont{\upshape}

\theoremstyle{nonumberplain}
\theoremheaderfont{\scshape}
\theorembodyfont{\normalfont}
\theoremsymbol{\ensuremath{\square}}

\theoremsymbol{\ensuremath{\lrcorner}}

\newcommand{\cT}{ \mathcal{T} }

\renewcommand{\phi}{\varphi}
\renewcommand{\epsilon}{\varepsilon}
\newcommand{\Oh}{\mathcal{O}}
\newcommand{\Ohs}{\mathcal{O}^*}

\renewcommand{\leq}{\leqslant}
\renewcommand{\geq}{\geqslant}

\newcommand{\homo}[1]{\textsc{Hom}(\ensuremath{#1})\xspace}
\newcommand{\lhomo}[1]{\textsc{LHom}(\ensuremath{#1})\xspace}
\newcommand{\coloring}[1]{\ensuremath{#1}-\textsc{Coloring}\xspace}

\newcommand{\id}{\operatorname{id}}
\newcommand{\tw}[1]{{\operatorname{tw}(#1)}}
\newcommand{\pw}[1]{{\operatorname{pw}(#1)}}
\newcommand{\og}{\operatorname{og}}
\newenvironment{inproof}{\noindent {\emph{Proof of Claim.}}}{\hfill$\blacksquare$\smallskip}

\begin{document}
\title{Fine-grained complexity of the graph homomorphism problem\\ for bounded-treewidth graphs\thanks{
The extended abstract of this work was presented during the conference SODA 2020~\cite{DBLP:conf/soda/OkrasaR20}}}


\author[1,2]{Karolina Okrasa\thanks{E-mail: \texttt{k.okrasa@mini.pw.edu.pl}.   Supported by the ERC grant CUTACOMBS (no. 714704).}}
\author[1,2]{Pawe\l{}~Rz\k{a}\.zewski\thanks{E-mail: \texttt{p.rzazewski@mini.pw.edu.pl}. Supported by Polish National Science Centre grant no. 2018/31/D/ST6/00062.}}

\affil[1]{Faculty of Mathematics and Information Science, Warsaw University of Technology, Poland}
\affil[2]{Faculty of Mathematics, Informatics and Mechanics, University of Warsaw}

\begin{titlepage}
\def\thepage{}
\thispagestyle{empty}
\maketitle

\begin{abstract}
For graphs $G$ and $H$, a \emph{homomorphism} from $G$ to $H$ is an edge-preserving
mapping from the vertex set of $G$ to the vertex set of $H$.
For a fixed graph $H$, by \textsc{Hom($H$)} we denote the computational problem
which asks whether a given graph $G$ admits a homomorphism to $H$.
If $H$ is a complete graph with $k$ vertices, then  \textsc{Hom($H$)} is equivalent 
to the $k$-\textsc{Coloring} problem, so graph homomorphisms can be seen as generalizations
of colorings.
It is known that \textsc{Hom($H$)} is polynomial-time solvable if $H$ is bipartite
or has a vertex with a loop, and NP-complete otherwise [Hell and Ne\v{s}et\v{r}il,
JCTB 1990].

In this paper we are interested in the complexity of the problem, parameterized
by the treewidth of the input graph $G$. If $G$ has $n$ vertices and is given along
with its tree decomposition of width $\mathrm{tw}(G)$, then the problem can be solved
in time $|V(H)|^{\mathrm{tw}(G)} \cdot n^{\mathcal{O}(1)}$, using a straightforward
dynamic programming. We explore whether this bound can be improved.
We show that if $H$ is a \emph{projective core}, then the existence of such a faster
algorithm is unlikely: assuming the Strong Exponential Time Hypothesis (SETH),
the \textsc{Hom($H$)} problem cannot be solved in time
$(|V(H)|-\epsilon)^{\mathrm{tw}(G)} \cdot n^{\mathcal{O}(1)}$, for any $\epsilon > 0$.
This result provides a full complexity characterization for a large class of graphs $H$,
as almost all graphs are projective cores.

We also notice that the naive algorithm can be improved for some graphs $H$, and show a 
complexity classification for all graphs $H$, assuming two conjectures from algebraic 
graph theory. In particular, there are no known graphs $H$ which are not covered by our
result.

In order to prove our results, we bring together some tools and techniques from algebra
and from fine-grained complexity.
\end{abstract}
\end{titlepage}

\section{Introduction}
Many problems that are intractable for general graphs become significantly easier if the structure of the input instance is ``simple''. One of the most successful measures of such a structural simplicity is the \emph{treewidth} of a graph, whose notion was rediscovered by many authors in different contexts \cite{bertele1972nonserial,Halin1976,DBLP:journals/jct/RobertsonS84,Arnborg1987}.
Most classic NP-hard problems, including \textsc{Independent Set}, \textsc{Dominating Set}, \textsc{Hamiltonian Cycle}, or \textsc{Coloring}, can be solved in time $\Ohs(f(\tw{G}))$, where $\tw{G}$ is the treewidth of the input graph $G$ (in the $\Ohs(\cdot)$ notation we suppress factors polynomial in the input size)~\cite{DBLP:journals/dam/ArnborgP89,DBLP:journals/cj/BodlaenderK08,DBLP:journals/iandc/Courcelle90,cygan2015parameterized}. In other words, many problems become polynomially solvable for graphs with bounded treewidth.

In the past few years the notion of \emph{fine-grained complexity} gained popularity, and the researchers became interested in understanding what is the optimal dependence on the treewidth, i.e., the function $f$ in the complexity of algorithms solving particular problems. This led to many interesting algorithmic results and lower bounds~\cite{DBLP:conf/esa/RooijBR09,DBLP:journals/iandc/BodlaenderCKN15,DBLP:journals/corr/KociumakaP17,DBLP:journals/eatcs/LokshtanovMS11,10.1007/978-3-642-22993-0_47,cygan2015parameterized}. Note that the usual assumption that P $\neq$ NP is not strong enough to obtain tight bounds for the running times of algorithms. In the negative results we usually assume the Exponential Time Hypothesis (ETH), or the Strong Exponential Time Hypothesis (SETH)~\cite{DBLP:journals/jcss/ImpagliazzoP01,DBLP:journals/jcss/ImpagliazzoPZ01}. Informally speaking, the ETH asserts that 3-\textsc{Sat} with $n$ variables and $m$ clauses cannot be solved in time $2^{o(n+m)}$, while the SETH implies that \textsc{CNF-Sat} with $n$ variables and $m$ clauses cannot be solved in time $(2 - \epsilon)^n \cdot m^{\Oh(1)}$, for any $\epsilon > 0$.

For example, it is known that for every fixed $k$, the \coloring{k} problem can be solved in time $\Ohs(k^{\tw{G}})$, if a tree decomposition of $G$ of width $\tw{G}$ is given~\cite{bodlaender2013fine, cygan2015parameterized}. On the other hand, Lokshtanov, Marx, and Saurabh showed that this result is essentially optimal, assuming the SETH.

\begin{theorem}[Lokshtanov, Marx, Saurabh~\cite{DBLP:journals/talg/LokshtanovMS18}] \label{thm:LMS}
Let $k \geq 3$ be a fixed integer. Assuming the SETH, the \coloring{k} problem on a graph $G$ cannot be solved in time $\Oh^*\left((k-\epsilon)^{\tw{G}}\right)$ for any $\epsilon>0$.
\end{theorem}

\paragraph{Homomorphisms} 
For two graphs $G$ and $H$, a homomorphism is an edge-preserving mapping from $V(G)$ to $V(H)$. The graph $H$ is called the \emph{target} of the homomorphism.
The existence of a homomorphism from any graph $G$ to the complete graph $K_k$ is equivalent to the existence of a $k$-coloring of $G$. Because of that we often refer to a homomorphism to $H$ as an \emph{$H$-coloring} and think of vertices of $H$ as colors. We also say that a graph $G$ is \emph{$H$-colorable} if it admits a homomorphism to $H$.
For a fixed graph $H$, by \homo{H} we denote the computational problem which asks whether a given instance graph $G$ admits a homomorphism to $H$. Clearly \homo{K_k} is equivalent to \coloring{k}.

Since \coloring{k} is arguably one of the best studied computational problems, it is interesting to investigate how these results generalize to \homo{H} for non-complete targets $H$. For example, it is known that \coloring{k} is polynomial-time solvable for $k \leq 2$, and NP-complete otherwise. A celebrated result by Hell and Ne\v{s}et\v{r}il~\cite{DBLP:journals/jct/HellN90} states that \homo{H} is polynomially solvable if $H$ is bipartite or has a vertex with a loop, and otherwise is NP-complete. The polynomial part of the theorem is straightforward and the main contribution was to prove hardness for all non-bipartite graphs $H$.
The difficulty comes from the fact that the local structure of the graph $H$ is not very helpful, but we need to consider $H$ as a whole. This is the reason why the proof of Hell and Ne\v{s}et\v{r}il uses a combination of combinatorial and algebraic arguments. Several alternative proofs of the result have appeared~\cite{DBLP:journals/tcs/Bulatov05,Siggers}, but none of them is purely combinatorial.

When it comes to the running times of algorithms for \coloring{k}, it is well-known that the trivial $\Ohs(k^n)$ algorithm for \coloring{k}, where $n$ is the number of vertices of the input graph, can be improved to $\Ohs(c^n)$ for a constant $c$ which does not depend on $k$ (currently the best algorithm of this type has running time $\Ohs(2^n)$~\cite{DBLP:journals/siamcomp/BjorklundHK09}).
Analogously, we can ask whether the trivial $\Ohs(|H|^n)$ algorithm for \homo{H} can be improved, where by $|H|$ we mean the number of vertices of $H$.
There are several algorithms with running times $\Ohs({c(H)}^n)$, where $c(H)$ is some structural parameter of $H$, which could be much smaller than $|H|$~\cite{DBLP:journals/mst/FominHK07,DBLP:journals/mst/Wahlstrom11,DBLP:journals/ipl/Rzazewski14}. However, the question whether there exists an absolute constant $c$, such that for every $H$ the \homo{H} problem can be solved in time $\Ohs(c^n)$, remained open. Finally, it was answered in the negative by Cygan {\em et al.}~\cite{DBLP:journals/jacm/CyganFGKMPS17}, who proved that the $\Ohs(|H|^n)$ algorithm is essentially optimal, assuming the ETH.

Using a standard dynamic programming approach, \homo{H} can be solved in time $\Ohs(|H|^t)$, if an input graph is given along with its tree decomposition of width $t$~\cite{bodlaender2013fine,cygan2015parameterized}.
\cref{thm:LMS} asserts that this algorithm is optimal if $H$ is a complete graph with at least 3 vertices, unless the SETH fails. A natural extension of this result would be to provide analogous tight bounds for non-complete targets $H$.

Egri, Marx, and Rzążewski~\cite{DBLP:conf/stacs/EgriMR18} considered this problem in the setting of \emph{list homomorphisms}. Let $H$ be a fixed graph. The input of the \lhomo{H} problem consists of a graph $G$, whose every vertex is equipped with a \emph{list} of vertices of the target $H$. We ask if $G$ has a homomorphism to $H$, respecting the lists.
Egri \emph{et al.} provided a full complexity classification for the case if $H$ is reflexive, i.e., every vertex has a loop. It is perhaps worth mentioning that a P / NP-complete dichotomy for \lhomo{H} was first proved for reflexive graphs as well: If $H$ is a reflexive graph, then the \lhomo{H} problem is polynomial time-solvable if $H$ is an interval graph, and NP-complete otherwise~\cite{FEDER1998236}. Egri {\em et al.} defined a new graph invariant $i^*(H)$, based on incomparable sets of vertices, and a new graph decomposition, and proved the following.

\begin{theorem}[Egri, Marx, Rzążewski~\cite{DBLP:conf/stacs/EgriMR18}] \label{thm:lhomo}
Let $H$ be a fixed non-interval reflexive graph with $i^*(H) = k$. Let $t$ be the treewidth of an instance graph $G$.
\begin{compactenum}[(a)]
\item Assuming a tree decomposition of $G$ of width $t$ is given, the \lhomo{H} problem can be solved in time~$\Ohs(k^t)$.
\item There is no algorithm solving the \lhomo{H} problem in time $\Oh^*\left((k-\epsilon)^{t}\right)$ for any $\epsilon >0$, unless the SETH fails.
\end{compactenum}
\end{theorem}

In this paper we are interested in showing tight complexity bounds for the complexity of the non-list variant of the problem.
Let us point out that despite the obvious similarity of \homo{H} and \lhomo{H} problems, they behave very differently when it comes to showing hardness results. Note that if $H'$ is an induced subgraph of $H$, then any instance of \lhomo{H'} is also an instance of \lhomo{H}, where the vertices of $V(H) \setminus V(H')$ do not appear in any list.
Thus in order to prove hardness of \lhomo{H}, it is sufficient to find a ``hard part'' $H'$ of $H$, and  perform a reduction for the \lhomo{H'} problem.
The complexity dichotomy for \lhomo{H} was proven exactly along these lines~\cite{FEDER1998236,DBLP:journals/combinatorica/FederHH99,DBLP:journals/jgt/FederHH03}. Also the proof of \cref{thm:lhomo} (b) heavily uses the fact that we can work with some local subgraphs of $H$ and ignore the rest of vertices. In particular, all these proofs are purely combinatorial.

On the other hand, in the \homo{H} problem, we need to capture the structure of the whole graph $H$, which is difficult using only combinatorial tools. This is why typical tools used in this area come from abstract algebra and algebraic graph theory.

For more information about graph homomorphisms we refer the reader to the comprehensive monograph by Hell and Ne\v{s}et\v{r}il~\cite{hell2004graphs}.

\paragraph{Our contribution} It is well known that in the study of graph homomorphisms the crucial role is played by the graphs that are \emph{cores}, i.e., they do not have a homomorphism to any of its proper subgraphs. In particular, in order to provide a complete complexity classification of \homo{H}, it is sufficient to consider the case that $H$ is a {connected} core (we explain this in more detail in \cref{sec:algo}). Also, the complexity dichotomy by Hell and Ne\v{s}et\v{r}il~\cite{DBLP:journals/jct/HellN90} implies that \homo{H} is polynomial-time solvable if $H$ is a graph on at most two vertices. So from now on let us assume that $H$ is a fixed core which is non-trivial, i.e., has at least three vertices.

We split the analysis into two cases, depending on the structure of $H$. First, in \cref{sec:proj}, we consider targets $H$ that are \emph{projective} (the definition of this class is rather technical, so we postpone it to \cref{sec:pre-prod}).
We show that for projective cores the straightforward dynamic programming on a tree decomposition is optimal, assuming the SETH.

\begin{restatable}{theorem}{projective}\label{thm:projective}
Let $H$ be a non-trivial projective core on $k$ vertices, and let $n$ and $t$ be, respectively, the number of vertices and the treewidth of an instance graph $G$.
\begin{compactenum}[(a)]
\item Even if $H$ is given as a input, the \homo{H} problem can be solved in time $\Oh(k^4+k^{t+1} \cdot n)$, assuming a tree decomposition of $G$ of width $t$ is given.
\item Even if $H$ is fixed, there is no algorithm solving the \homo{H} problem in time $\Oh^*\left((k-\epsilon)^{t}\right)$ for any $\epsilon >0$, unless the SETH fails.
\end{compactenum}
\end{restatable}

The proof brings together some tools and ideas from algebra and fine-grained complexity theory. The main technical ingredient is the construction of a so-called \emph{edge gadget}, i.e., a graph $F$ with two specified vertices $u^*$ and $v^*$, such that:
\begin{compactenum}[(a)]
\item for any distinct vertices $x,y$ of $H$, there is a homomorphism from $F$ to $H$, which maps $u^*$ to $x$ and $v^*$ to $y$, and
\item in any homomorphism from $F$ to $H$, the vertices $u^*$ and $v^*$ are mapped to distinct vertices of $H$.
\end{compactenum}

Using this gadget, we can perform a simple and elegant reduction from \coloring{k}. If $G$ is an instance of \coloring{k}, we construct an instance $G^*$ of \homo{H} by taking a copy of $G$ and replacing each edge $xy$ with a copy of the edge gadget, whose $u^*$-vertex is identified with $x$, and $v^*$-vertex is identified with $y$.
By the properties of the edge gadget it is straightforward to observe that $G^*$ is $H$-colorable if and only if $G$ is $k$-colorable. Since the size of $F$ depends only on $H$, we observe that the treewidth of $G^*$ differs from the treewidth of $G$ by an additive constant, which is sufficient to obtain the desired lower bound.

Although the statement of \cref{thm:projective} might seem quite specific, it actually covers a large class of graphs.
We say that a property $P$ holds for \emph{almost all graphs}, if the probability that a graph chosen at random from the family of all graphs with vertex set $\{1,2,\ldots,n\}$ satisfies $P$ tends to 1 as $n \to \infty$.
Hell and Ne\v{s}et\v{r}il observed that almost all graphs are cores~\cite{DBLP:journals/dm/HellN92}, see also \cite[Corollary 3.28]{hell2004graphs}. Moreover, Łuczak and Ne\v{s}et\v{r}il proved that almost all graphs are projective~\cite{luczak2004note}. From these two results, we can  obtain that almost all graphs are projective cores. This, combined by \cref{thm:projective}, implies the following.

\begin{corollary}
For almost all graphs $H$, the \homo{H} problem on instance graphs with treewidth~$t$ cannot be solved in time $\Oh^*\left((|H|-\epsilon)^{t}\right)$ for any $\epsilon >0$, unless the SETH fails.
\end{corollary}

In \cref{sec:nonproj} we consider the case that $H$ is a non-projective core. First, we show that the approach that we used for projective cores cannot work in this case: it appears that one can construct the edge gadget for a core $H$ with the properties listed above if and only if $H$ is projective.
What makes studying non-projective cores difficult is that we do not understand their structure well.
In particular, we know that a graph $H = H_1 \times H_2$, where $H_1$ and $H_2$ are non-trivial and $\times$ denotes the \emph{direct product} of graphs (see \cref{sec:pre-prod} for a formal definition), is non-projective, and by choosing $H_1$ and $H_2$ appropriately, we can ensure that $H$ is a core.
However, we do not know whether there are any non-projective non-trivial connected cores that are \emph{indecomposable}, i.e., they cannot be constructed using direct products.
This problem was studied in a slightly more general setting by Larose and Tardif~\cite[Problem 2]{larose2001strongly}, and it remains wide open.
We restate it here, only for restricted case that $H$ is a core, which is sufficient for our purpose.

\begin{restatable}{conjecture}{conprojprime}
\label{con:proj-iff-prime}
Let $H$ be a connected non-trivial core. Then $H$ is projective if and only if it is indecomposable.
\end{restatable}

Since we do not know any counterexample to \cref{con:proj-iff-prime}, in the remainder we consider cores $H$ that are built using the direct product.
If $H= H_1 \times \ldots \times H_m$ and each $H_i$ is non-trivial and indecomposable, we call $H_1 \times \ldots \times H_m$ a \emph{prime factorization} of $H$. For such $H$ we show a lower complexity bound for \homo{H}, under an additional assumption that one of the factors $H_i$ of $H$ is \emph{truly projective}.

The definition of truly projective graphs is rather technical and we present it in \cref{sec:nonproj}. 
Graphs with such a property (actually, a slightly more restrictive one) were studied by Larose~\cite[Problems 1b. and 1b'.]{Larose2002FamiliesOS} in connection with some problems related to unique colorings, considered by Greenwell and Lov{\'a}sz~\cite{greenwell1974applications}. Larose~\cite{Larose2002FamiliesOS,larose2002strongly} defined and investigated even more restricted class of graphs, called \emph{strongly projective} (see \cref{sec:conclusion} for the definition). We know that every strongly projective graph is truly projective, and every truly projective graph is projective.
Larose~\cite{Larose2002FamiliesOS,larose2002strongly} proved that all known projective graphs are in fact strongly projective. This raises a natural question whether projectivity and strong projectivity are in fact equivalent~\cite{Larose2002FamiliesOS,larose2002strongly}. Of course, an affirmative answer to this question would in particular mean that all projective cores are truly projective. Again, we state the problem in this weaker form, which is sufficient for our application.

\begin{restatable}{conjecture}{conjstrongly}\label{con:strongly}
Every projective core is truly projective.
\end{restatable}

Actually, if we assume both \cref{con:proj-iff-prime} and \cref{con:strongly}, we are able to provide a full complexity classification for the \homo{H} problem, parameterized by the treewidth of the input graph.

\begin{restatable}{theorem}{thmeverything}\label{thm:everything}
Assume that \cref{con:proj-iff-prime} and \cref{con:strongly} hold. 
Let $H$ be a non-trivial connected core with prime factorization $H_1 \times \ldots \times H_m$, and define $k:= \max_{i \in [m]}|H_i|$. Let $n$ and $t$ be, respectively, the number of vertices and the treewidth of an instance graph $G$. \begin{compactenum}[(a)]
\item Even if $H$ is given as an input, the \homo{H} problem can be solved in time $\Oh(|H|^4+k^{t+1}\cdot n)$, assuming a tree decomposition of $G$ of width $t$ is given.
\item Even if $H$ is fixed, there is no algorithm solving the \homo{H} problem in time $\Oh^*\left((k-\epsilon)^t\right)$ for any $\epsilon >0$, unless the SETH fails.
\end{compactenum}
\end{restatable}

Let us point out that despite some work on both conjectures~\cite{larose2001strongly,Larose2002FamiliesOS,larose2002strongly}, we know no graph $H$ for which the bounds from \cref{thm:everything} do not hold.

\section{Notation and preliminaries} \label{sec:prelim}
For $n \in \mathbb{N}$, we denote the set $\{1,2,\ldots,n\}$ by~$[n]$.
All graphs considered in this paper are finite, undirected and do not contain parallel edges.
For a graph $G$, by $V(G)$ and $E(G)$ we denote the set of vertices and the set of edges of $G$, respectively, and we write $|G|$ for the number of vertices of $G$.
Let $K_1^*$ be the single-vertex graph with a loop. 
A graph is \emph{ramified} if it has no two distinct vertices $u$ and $v$ such that the open neighborhood of $u$ is contained in the open neighborhood of $v$. An \emph{odd girth} of a graph $G$, denoted by $\og(G)$, is the length of a shortest odd cycle in $G$. For a graph $G$, denote by $\omega(G)$ and $\chi(G)$, respectively, the size of the largest clique contained in $G$ and the chromatic number of $G$.

A \emph{tree decomposition} of a graph $G$ is a pair $\left(\cT, \{X_a\}_{a \in V(\cT)}\right)$, in which $\cT$ is a tree, whose vertices are called \emph{nodes} and $\{X_a\}_{a \in V(\cT)}$ is the family of subsets (called \emph{bags}) of $V(G)$, such that
\begin{compactenum}
\item every $v \in V(G)$ belongs to at least one bag $X_a$,
\item for every $uv \in E(G)$ there is at least one bag $X_a$ such that $u,v \in X_a$,
\item for every $v \in V(G)$ the set $\cT_v:=\{a \in V(\cT)~|~ v \in X_a\}$ induces a connected subgraph of~$\cT$.
\end{compactenum}
The \emph{width} of a tree decomposition $\left(\cT, \{X_a\}_{a \in V(\cT)}\right)$ is the number $\max_{a\in V(\cT)}|X_a|-1$. The minimum possible width of a tree decomposition of $G$ is called the \emph{treewidth} of $G$ and denoted by $\tw{G}$.
In particular, if $\cT$ is a path, then a tree decomposition $\left(\cT, \{X_a\}_{a \in V(\cT)}\right)$ is called a \emph{path decomposition}. The \emph{pathwidth} of $G$, denoted by $\pw{G}$, is the minimum possible width of a path decomposition of $G$. Clearly for every graph $G$ it holds that $\tw{G} \leq \pw{G}$.

\subsection{Graph homomorphisms and cores} \label{sec:pre-homo}
For graphs $G$ and $H$, a function $f: V(G) \to V(H)$ is a \emph{homomorphism} if it preserves edges, i.e., for every $uv \in E(G)$ it holds that $f(u)f(v) \in E(H)$ (see \cref{fig:homo}).
If $G$ admits a homomorphism to $H$, we denote this fact by $G \to H$ and we write $f:G \to H$ if $f$ is a homomorphism from $G$ to $H$. If there is no homomorphism from $G$ to $H$, we write $G \not\to H$. 
Graphs $G$ and $H$ are \emph{homomorphically equivalent} if $G \to H$ and $H \to G$, and \emph{incomparable} if $G \not\to H$ and $H \not\to G$. Observe that homomorphic equivalence is an equivalence relation on the class of all graphs. An \emph{endomorphism} of $G$ is any homomorphism $f \colon G \to G$.

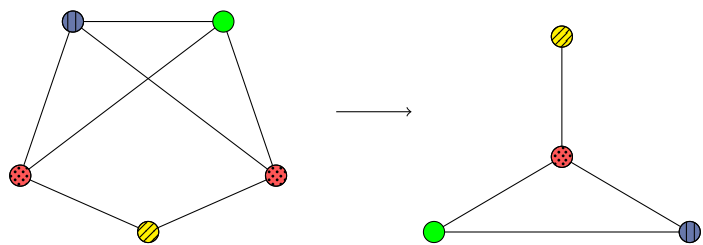
\begin{figure}[ht]
\centering{\begin{tikzpicture}[every node/.style={draw,circle,fill=white,inner sep=0pt,minimum size=8pt},every loop/.style={}]
\draw (2,0) -- (0.3,.75)-- (1,2.8) -- (3, 2.8)  -- (3.7,.75)  -- (2,0);
\draw (0.3,.75) -- (3, 2.8);
\draw (3.7,.75) -- (1, 2.8);
\node[fill=green] at (3,2.8) {};
\draw[->] (4.5,1.6) -- (5.5,1.6);
\draw (5.8,0) -- (9.2,0) -- (7.5,1) -- (5.8,0);
\draw (7.5,1) -- (7.5,2.6);
\node[fill=yellow] at (2,0) {}; \node[pattern=north east lines] at (2,0) {}; 
\node[fill=red!66] at (0.3,0.75) {}; \node[pattern=crosshatch dots] at (0.3,0.75) {}; 
\node[fill=red!66] at (7.5,1) {}; \node[pattern=crosshatch dots] at (7.5,1) {}; 
\node[fill=red!66] at (3.7,0.75) {}; \node[pattern=crosshatch dots] at (3.7,0.75) {}; 
\node[fill=green] at (5.8,0) {}; 
\node[fill=blue!66] at (1,2.8) {}; \node[pattern=vertical lines] at (1,2.8) {}; 
\node[fill=blue!66] at (9.2,0) {}; \node[pattern=vertical lines] at (9.2,0) {}; 
\node[fill=yellow] at (7.5,2.6) {}; \node[pattern=north east lines] at (7.5,2.6) {};
\end{tikzpicture}}
\caption{An example of a homomorphism from $G$ (left) to $H$ (right).  Patterns on the vertices indicate the mapping.}
\label{fig:homo}
\end{figure}

A graph $G$ is a \emph{core} if $G \not\to H$ for every proper subgraph $H$ of $G$.
Equivalently, we can say $G$ is a core if and only if every endomorphism of $G$ is an automorphism (i.e., an isomorphism from $G$ to $G$).
Note that a core is always ramified.
If $H$ is a subgraph of $G$ such that $G \to H$ and $H$ is a core, we say that $H$ is a core of~$G$. Notice that if $H$ is a subgraph of $G$, then it always holds that $H \to G$, so every graph is homomorphically equivalent to its core. Moreover, if $H$ is a core of $G$, then $H$ is always an induced subgraph of $G$, because every endomorphism $f:G \to H$ restricted to $H$ must be an automorphism. It was observed by Hell and Ne\v{s}et\v{r}il that every graph has a unique core (up to an isomorphism)~\cite{DBLP:journals/dm/HellN92}. 
Note that if $f:G \to H$ is a homomorphism from $G$ to its core $H$, then it must be surjective. 

We say that a core is \emph{trivial} if it is isomorphic to $K_1$, $K_1^*$, or $K_2$. It is easy to observe that these three graphs are the only cores with fewer than 3 vertices. In general, finding a core of a given graph is computationally hard; in particular, deciding if a graph is a core is coNP-complete~\cite{DBLP:journals/dm/HellN92}. However, the graphs whose cores are trivial are simple to describe.

\begin{observation}\label{obs:trivial}
Let $G$ be a graph, whose core $H$ is trivial.
\begin{compactenum}[(a)]
\item $H \simeq K_1$ if and only if $\chi(G)=1$, i.e., $G$ has no edges,
\item $H \simeq K_2$ if and only if $\chi(G)=2$, i.e., $G$ is bipartite and has at least one edge,
\item $H \simeq K_1^*$ if and only if $G$ has a vertex with a loop.{\hfill$\square$\smallskip}
\end{compactenum}
\end{observation}

In particular, there are no non-trivial cores with loops. The following conditions are necessary for $G$ to have a homomorphism into $H$.
\begin{observation}[\cite{hell2004graphs}]\label{obs:incompcores}
Assume that $G \to H$ and $G$ and $H$ have no loops. Then $\omega(G) \leq \omega(H)$, $\chi(G) \leq \chi(H)$, and $\og(G)\geq \og(H)$. {\hfill$\square$\smallskip}
\end{observation}

We denote by $H_1 + \ldots + H_m$ a disconnected graph with connected components $H_1, \dots, H_m$.
Observe that if $f$ is a homomorphism from $G=G_1 + \ldots + G_\ell$ to $H=H_1 + \ldots + H_m$, then it maps every connected component of $G$ into some connected component of $H$. Also note that a graph does not have to be connected to be a core, in particular the following characterization follows directly from the definition of a core.

\begin{observation} \label{obs:components}
A disconnected graph $H$ is a core if and only if its connected components are pairwise incomparable cores.{\hfill$\square$\smallskip}
\end{observation}

An example of a pair of incomparable cores is shown in  \cref{fig:grotzsch}: it is the \emph{Grötzsch graph}, denoted by $G_G$, and the clique $K_3$. Clearly, $\og(G_G) > \og(K_3)$ and $\chi(G_G) > \chi(K_3)$, so by \cref{obs:incompcores}, they are incomparable. Therefore, by \cref{obs:components}, the graph $G_G + K_3$ is a core. 

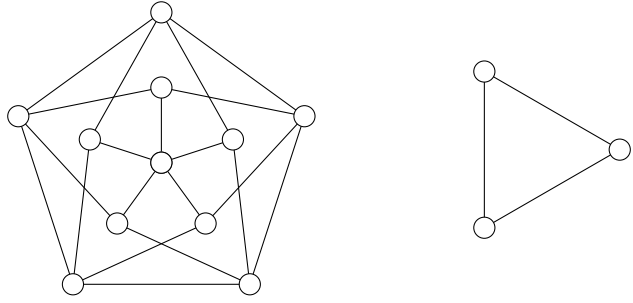
\begin{figure}[ht]
\centering{\begin{tikzpicture}[every node/.style={draw,circle,fill=white,inner sep=0pt,minimum size=8pt},every loop/.style={}, rotate = 360/20]
\def\n{5}
\node (a) at (0,0) {};
\foreach \i in {1,...,\n}
{
	\draw (360/\n*\i-360/\n:2) -- (360/\n*\i:2);
	\draw (360/\n*\i-360/\n:1) -- (360/\n*\i:2);
	\draw (360/\n*\i+360/\n:1) -- (360/\n*\i:2);
	\draw (360/\n*\i-360/\n:1) -- (a);
}
\foreach \i in {1,...,\n}
{
	\node (a\i) at (360/\n*\i:2) {};	
	\node (a\i) at (360/\n*\i:1) {};
}
\node (a) at (0,0) {};
\end{tikzpicture} \hskip 2cm
\begin{tikzpicture}[every node/.style={draw,circle,fill=white,inner sep=0pt,minimum size=8pt},every loop/.style={}]
\def\n{3}
\foreach \i in {1,...,\n}
{
	\draw (360/\n*\i-360/\n:1.2) -- (360/\n*\i:1.2);
}
\foreach \i in {1,...,\n}
{
	\node (a\i) at (360/\n*\i:1.2) {};	
}
\node[fill=none,draw=none] at (0,-1.8) {};
\end{tikzpicture}}
\caption{An example of incomparable cores, the \emph{Grötzsch graph} (left) and $K_3$ (right).}
\label{fig:grotzsch}
\end{figure}

Finally, let us observe that we can construct arbitrarily large families of pairwise incomparable cores. Let us start the construction with an arbitrary non-trivial core $H_0$. Now suppose we have constructed pairwise incomparable cores $H_0,H_1,\ldots,H_k,$ and we want to construct $H_{k+1}$. Let $\ell= \max_{i \in \{0,\ldots, k\}} \og(H_i)$, $r= \max_{i \in \{0,\ldots, k\}} \chi(H_i)$. By the classic result of Erd\H{o}s \cite{erdos1959graph}, there is a graph $H$ with $\og(H)>\ell$ and $\chi(H)>r$.
We set $H_{k+1}$ to be the core of $H$. Observe that $\og(H_{k+1}) = \og(H) > \ell$ and $\chi(H_{k+1}) = \chi(H) > r$, so, by \cref{obs:incompcores}, we have that for every $i \in \{0,\ldots, k\}$ the core $H_{k+1}$ is incomparable with $H_i$.

\subsection{Graph products and projectivity} \label{sec:pre-prod}
Define the \emph{direct product} of graphs $H_1$ and $H_2$, denoted by $H_1 \times H_2$, as follows:
\begin{align*}
V(H_1 \times H_2)  = & \{(x,y)~|~ x \in V(H_1)\textrm{ and }y \in V(H_2)\} \textrm{ and }\\
 E(H_1 \times H_2) = & \{(x_1,y_1)(x_2,y_2)~|~ x_1x_2 \in E(H_1) \textrm{ and } y_1y_2 \in E(H_2)\}.
\end{align*}
If $H=H_1 \times H_2$, then  $H_1 \times H_2$ is a \emph{factorization} of $H$, and $H_1$ and $H_2$ are its  \emph{factors}.
Clearly, the binary operation $\times$ is commutative, so will identify $H_1 \times H_2$ and $H_2 \times H_1$.
Since $\times$ is also associative, we can extend the definition for more than two factors: 
\[H_1 \times \dots \times H_{m-1} \times H_m := (H_1 \times \dots \times H_{m-1}) \times H_m.\] 
Moreover, in the next sections, we will sometimes consider products of graphs, that are products themselves. Formally, the vertices of such graphs are tuples of tuples.
If it does not lead to confusion, for $\bar{x}:=(x_1,\dots, x_{k_1})$ and $\bar{y}:=(y_1,\dots, y_{k_2})$, we will treat tuples $(\bar{x},\bar{y}),$ $(x_1,\dots, x_{k_1},y_1,\dots, y_{k_2})$, $(\bar{x},y_1,\dots, y_{k_2})$, and $(x_1,\dots, x_{k_1},\bar{y})$ as equivalent. This notation is generalized to more factors in a natural way. We denote by $H^m$ the product of $m$ copies of $H$.

The direct product appears in the literature under different names: \emph{tensor product}, \emph{cardinal product}, \emph{Kronecker product}, \emph{relational product}. It is also called \emph{categorical product}, because it is the product in the category of graphs (see \cite{hammack2011handbook,miller_1968} for details).

Note that if $H_1 \times H_2$ has at least one edge, then $H_1 \times H_2 \simeq H_1$ if and only if $H_2 \simeq K_1^*$.
We say that a graph $H$ is \emph{directly indecomposable} (or \emph{indecomposable} for short) if the fact that $H = H_1 \times H_2$ implies that either $H_1 \simeq K_1^*$ or $H_2 \simeq K_1^*$. A graph that is not indecomposable, is \emph{decomposable}. 
A factorization, where each factor is directly indecomposable and not isomorphic to $K_1^*$, is called a \emph{prime factorization}.  Clearly, $K_1^*$ does not have a prime factorization.

The following property will be very useful (see also Theorem 8.17 in~\cite{hammack2011handbook}).
\begin{theorem}[McKenzie~\cite{McKenzie1971}]\label{thm:unique-fact}
Any connected non-bipartite graph with more than one vertex has a unique prime factorization into directly indecomposable factors (with possible loops).
\end{theorem}

Let $H_1 \times \ldots \times H_m$ be some factorization of $H$ (not necessary prime) and let $i \in [m]$. A function $\pi_i:V(H) \to V(H_i)$ such that for every $(x_1,\dots,x_m) \in V(H)$ it holds that $\pi_i(x_1,\dots,x_m)=x_i$ is a \emph{projection on the $i$-th coordinate}. It follows from the definition of the direct product that every projection $\pi_i$ is a homomorphism from $H$ to $H_i$.

Below we summarize some basic properties of direct products.
\begin{observation}\label{obs:projprop}Let $H$ be a graph on $k$ vertices. Then
\begin{compactenum}[(a)]
\item \label{obs:kone} $H \times K_1$ consists of $k$ isolated vertices, in particular its core is $K_1$,
\item \label{obs:ktwo} if $H$ has at least one edge, then the core of $H \times K_2$ is $K_2$,
\item \label{obs:sub} the graph $H^m$ contains a subgraph isomorphic to $H$, which is induced by the set $\{(x,\dots,x)~|~ x \in V(H)\}$; in particular, if $m \geq 2$, then $H^m$ is never a core, 
\item \label{obs:connected} if $H=H_1 \times \ldots \times H_m$ and $H_1,H_2,\ldots,H_m$ are connected, then $H$ is connected if and only if at most one $H_i$ is bipartite,
\item \label{obs:homprod}if $H=H_1 \times \ldots \times H_m$, then for every $G$ it holds that $G \to H$ if and only if $G \to H_i$ for all $i \in [m]$.
\end{compactenum}
\end{observation}
\begin{proof} Items \eqref{obs:kone}, \eqref{obs:ktwo}, \eqref{obs:sub} are straightforward to observe. Item \eqref{obs:connected} follows from a result of Weichsel~\cite{weichsel1962kronecker}, see also \cite[Corollary 5.10]{hammack2011handbook}.
To prove \eqref{obs:homprod}, consider a homomorphism $f: G \to H$. Clearly, $H \to H_i$ for every $i \in [m]$ because each projection $\pi_i: H \to H_i$ is a homomorphism. So $\pi_i \circ f$ is a homomorphism from $G$ to $H_i$. On the other hand, if we have some $f_i:G \to H_i$ for every $i \in [m]$, then we can define a homomorphism $f: G \to H$ by $f(x):=(f_1(x),\dots,f_m(x))$.
\end{proof}  

A homomorphism $f:H^m \to H$ is \emph{idempotent}, if for every $x \in V(H)$ it holds that $f(x,x,\dots,x)=x$.
One of the main characters of the paper is the class of \emph{projective graphs}, considered e.g. in \cite{Larose2002FamiliesOS,larose2002strongly,larose2001strongly}.
A graph $H$ is \emph{projective} (or \emph{idempotent trivial}), if for every $m \geq 2$, every idempotent homomorphism from $H^m$ to $H$ is a projection. 

\begin{observation} \label{obs:composition}
If $H$ is a projective core and $f: H^m \to H$ is a homomorphism, then $f \equiv g \circ \pi_i$ for some $i \in [m]$ and some automorphism $g$ of $H$.
\end{observation}
\begin{proof}
If $f$ is idempotent, then it is a projection and we are done. Assume $f$ is not idempotent and define $g: V(H) \to V(H)$ by $g(x) = f(x, \dots, x)$. The function $g$ is an endomorphism of $H$ and $H$ is a core, so  $g$ is in fact an automorphism of $H$. Observe that $g^{-1} \circ f$ is an idempotent homomorphism, so it is equal to $\pi_i$ for some $i \in [m]$, because $H$ is projective. From this we get that $f \equiv g \circ \pi_i$.
\end{proof}

It is known that projective graphs are always connected~\cite{larose2001strongly}.
Observe that the definition of projective graphs does not imply that their recognition is decidable. However, an algorithm to recognize these graphs follows from the following, useful characterization.
\begin{theorem}[Larose, Tardif~\cite{larose2001strongly}]\label{thm:proj-char}
A connected graph $H$ with at least three vertices is projective if and only if every idempotent homomorphism from $H^2$ to $H$ is a projection. \end{theorem}

Recall from the introduction that almost all graphs are projective cores~\cite{hell2004graphs,luczak2004note}. 
It appears that the properties of projectivity and being a core are independent. In particular, the graph in \cref{fig:nocore} is not a core, as it can be mapped to a triangle. However, Larose~\cite{Larose2002FamiliesOS} proved that all non-bipartite, connected, ramified graphs which do not contain $C_4$ as a (non-necessarily induced) subgraph, are projective (this will be discussed in more detail in \cref{sec:conclusion}, see \cref{thm:squarefreee}).
On the other hand, there are also non-projective cores, an example is $G_G \times K_3$, see \cref{fig:grotzsch}. 
We discuss such graphs in detail in \cref{sec:nonproj}.

\begin{figure}[ht]
\centering{\begin{tikzpicture}[scale=0.9,every node/.style={draw,circle,fill=white,inner sep=0pt,minimum size=8pt},every loop/.style={}]
\draw (2,1) -- (0,0) -- (0,2) -- (2,1) -- (4,0) -- (4,2) -- (2,1);
\node at (0,0) {};
\node at (4,0) {};
\node at (0,2) {};
\node at (4,2) {};
\node at (2,1) {};
\end{tikzpicture}}
\caption{An example of a projective graph which is not a core.}
\label{fig:nocore}
\end{figure}
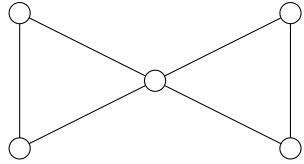

\section{Complexity of finding graph homomorphisms} \label{sec:algo}
Note that if two graphs $H_1$ and $H_2$ are homomorphically equivalent, then the \homo{H_1} and \homo{H_2} problems are also equivalent. So in particular, because every graph is homomorphically equivalent to its core, we may restrict our attention to graphs $H$ which are cores.
Also, recall from \cref{obs:trivial} that \homo{H} can be solved in polynomial time if $H$ is isomorphic to $K_1^*$, $K_1$, or $K_2$. So we will be interested only in non-trivial cores $H$. In particular, we will assume that $H$ is non-bipartite and has no loops.

We are interested in understanding the complexity bound of the \homo{H} problem, parameterized by the treewidth of the input graph.
The dynamic programming approach (see Bodlaender et al.~\cite{bodlaender2013fine}) gives us the following upper bound.


\begin{theorem}[Bodlaender et al.~\cite{bodlaender2013fine}]\label{thm:algo}
Let $H$ be a graph on $k$ vertices. Even if $H$ is a part of the input, the \homo{H} problem can be solved in time $\Oh(k^{t+1} \cdot n)$, assuming a tree decomposition of width $t$ of the instance graph on $n$ vertices is given.
\end{theorem}

By \cref{thm:LMS}, this bound is tight (up to the polynomial factor) if $H$ is a complete graph with at least three vertices, unless the SETH fails. We are interested in extending this result for other graphs $H$.

First, let us observe that there are cores, for which the bound from \cref{thm:algo} can be improved.
Indeed, let $H$ be a decomposable core, isomorphic to $H_1 \times \ldots \times H_m$ (see discussion in \cref{sec:nonproj} for more about cores that are products.).
Recall from \cref{obs:projprop} {(\ref{obs:homprod})} that for every graph $G$ it holds that \[G \to H\textrm{ if and only if }G \to H_i\textrm{ for every }i \in [m].\]
So, given an instance $G$ of \homo{H}, we can call the algorithm from \cref{thm:algo} to solve \homo{H_i} for each $i \in [m]$ and return a positive answer if and only if we get a positive answer in each of the calls. 
Since Hammack and Imrich~\cite{hammack2009cartesian} presented an algorithm for finding the prime factorization of $H$ in time $\Oh(|H|^4)$, we obtain the following result.

\begin{theorem}\label{thm:non-prime}
Let $H$ be a core with prime factorization $H_1 \times \ldots \times H_m$. 
Even if $H$ is given as a part of the input, the \homo{H} problem can be solved in time $\Oh\left(|H|^4+\max_{j \in [m]}|H_j|^{t+1}\cdot n\right)$, assuming a tree decomposition of width $t$ of the instance graph with $n$ vertices is given. {\hfill$\square$\smallskip}
\end{theorem}

Let us conclude this section with two lemmas, which justify why we restrict our attention to  connected cores.
Assume $H=H_1+\ldots+H_m$ is a disconnected core.
The first lemma shows how to solve \homo{H} using an algorithm for finding homomorphisms to connected target graphs. 

\begin{lemma}\label{obs:disconnected-not-fixed}
Consider a disconnected core $H =H_1+\ldots + H_m$.
Assume that for every $i \in [m]$ the \homo{H_i} problem can be solved in time $\Oh(|H_i|^4 + n^d \cdot c(H_i)^t)$ for an instance $G$ with $n$ vertices, given along with its tree decomposition of width $t$, where $c$ is some function and $d$ is a constant.
Then, even in $H$ is given as an input, the \homo{H} problem can be solved in time $\Oh\left(|H|^4+n^d \cdot \max_{i \in [m]}c(H_i)^t\cdot |H|\right)$.
\end{lemma}
\begin{proof}
First, observe that if $G$ is disconnected, say $G = G_1 + \ldots + G_\ell$, then $G \to H$ if and only if $G_i \to H$ for every $i \in [\ell]$.
Also, any tree decomposition of $G$ can be easily transformed to a tree decomposition of $G_i$ of at most the same width.
It means that if the instance graph is disconnected, we can just consider the problem separately for each of its connected components. So we assume that $G$ is connected. Then $G \to H$ if and only if $G \to H_i$ for some $i \in [m]$. We find the connected components of $H$ in time $\Oh(|H|^2)$ and then solve \homo{H_i} for each $i \in [m]$ (for the same instance $G$) in time $\Oh(|H_i|^4 + n^d \cdot c(H_i)^t)$. We return a positive answer for \homo{H} if and only if we get a positive answer for at least one $i \in [m]$. The total complexity of this algorithm is $\Oh\left(|H|^2+|H_1|^4+ \ldots + |H_m|^4 +n^d \cdot (c(H_1)^t + \ldots + c(H_m)^t)\right)=\Oh\left(|H|^4+n^d \cdot \max_{i \in [m]}c(H_i)^t\cdot |H|\right), $ as $|H_1|^4+ \ldots + |H_m|^4 \leq (|H_1|+ \ldots + |H_m|)^4 = |H|^4$.
\end{proof}

The second lemma shows that even if we assume $H$ to be fixed, we cannot solve \homo{H} faster than we solve \homo{H_i}, for each connected component $H_i$ of $H$. We state it in a stronger version, parameterized by the pathwidth of an instance graph.

\begin{lemma}\label{lem:disconnected-fixed}
Let $H =H_1+\ldots + H_m$ be a fixed, disconnected core. Assume that the \homo{H} problem can be solved in time $\Ohs(\alpha^\pw{G})$ for an instance $G$ given along with its optimal path decomposition. Then for every $i \in [m]$ the \homo{H_i} problem can be solved in time $\Ohs(\alpha^{\pw{G}})$.
\end{lemma}
\begin{proof}
Again, we may restrict our attention only to connected instances, as otherwise we can solve the problem separately for each instance. Consider a connected instance $G$ of \homo{H_i} on $n$ vertices and pathwidth $t$. Let $V(H_i)=\{z_1,\dots,z_k\}$ and let $u$ be some fixed vertex of $G$. 
We construct an instance $G^*$ of \homo{H} as follows. We take a copy $G'$ of $G$ and a copy $\widetilde{H}_i^k$ of $H_i^k$, and identify the vertex corresponding to $u$ in $G'$ with the vertex corresponding to $(z_1,\dots,z_k)$ in $\widetilde{H}_i^k$. Denote this vertex of $G^*$ by $\bar{z}$. Observe that $H_i$ a connected, non-trivial core, so \cref{obs:projprop} \eqref{obs:connected} implies that $\widetilde{H}_i^k$ is connected. Since $G$ is also connected, $G^*$ must be connected. 

We claim that $G \to H_i$ if and only if $G^* \to H$. Indeed, if $f: G \to H_i$, then there exists $j \in [k]$ such that $f(u)=z_j$, so we can define a homomorphism $g: G^* \to H_i$ (which is also a homomorphism from $G^*$ to $H$) by 
\[g(x) = \begin{cases}
f(x) & \textrm{if $x \in V(G'),$}\\
\pi_j(x) & \textrm{otherwise.}
\end{cases}
\]
Clearly, both $f$ and $\pi_j$ are homomorphisms. Recall that $\bar{z}$ is a cutvertex in $G^*$ obtained by identifying $u$ from $G'$ and  $(z_1,\dots,z_k)$ from $\widetilde{H}_i^k$. Furthermore, we have $g(\bar{z})= f(u)=\pi_j(z_1,\dots,z_k)=z_j$, so $g$ is a homomorphism from $G^*$ to $H$.

Conversely, if we have $g:G^* \to H$, we know that $g$ maps $G^*$ to a connected component $H_j$ of $H$, for some $j \in [m]$, because $G^*$ is connected. But $G^*$ contains an induced copy $\widetilde{H}_i^k$ of $H_i^k$, so also an induced copy of $H_i$, say $\widetilde{H}_i$ (recall \cref{obs:projprop} {(\ref{obs:sub})}). So $g|_{V(\widetilde{H}_i)}$ is in fact a homomorphism from $H_i$ to $H_j$. Recall from \cref{obs:incompcores} that since $H_1 +\ldots + H_m$ is a core, its connected components are pairwise incomparable cores -- so $j$ must be equal to $i$. It means that $g|_{V(G')}$ is a homomorphism from $G'$ to $H_i$, so we conclude that $G \to H_i$. 

Note that the number of vertices of $G^*$ is $n + |H_i^k| -1 \leq |H_i^k| \cdot n$.
Now let $\left(\cT, \{X_a\}_{a \in V(\cT)}\right)$ be a path decomposition of $G$ of width $t$, and let $b$ be a node of $\cT$, such that $u \in X_b$.
Let $\cT^*$ be the path obtained from $\cT$ by inserting a new node $b'$ as the direct successor of $b$.
Define $X_{b'} := X_b \cup V(H_i^k)$. Clearly, $\left( \cT^*, \{X_a\}_{a \in V(\cT^*)}\right)$ is a path decomposition of $G^*$. This means that $\pw{G^*}\leq t+|H_i^k|$. The graph $H_i$ is fixed, so the number of vertices of $H_i^k$ is a constant.
By our assumption we can decide if $G^* \to H$ in time $\alpha^\pw{G^*}\cdot c \cdot |G^*|^d$, so we can decide if $G \to H_i$ in time $\alpha^\pw{G^*}\cdot c \cdot (|H_i^k|n)^{d} \leq \alpha^t \alpha^{|H_i^k|}\cdot c \cdot |H_i^k|^{d} n^{d}=\alpha^t \cdot c' \cdot n^{d}$, where $c'=c \cdot \alpha^{|H_i^k|} \cdot |H_i^k|^{d}$.
\end{proof}

Let us point out that the assumptions in \cref{obs:disconnected-not-fixed} (that $H$ is given as an input) and \cref{lem:disconnected-fixed} (that $H$ is fixed) correspond, respectively, to the assumptions in statements (a) and (b) of \cref{thm:projective} and \cref{thm:everything}.

\section{Lower bounds}
In this section we will investigate the lower bounds for the complexity of \homo{H}. The section is split into two main parts. In \cref{sec:proj} we consider projective cores. Then, in \cref{sec:nonproj}, we consider non-projective cores.

\subsection{Projective cores} \label{sec:proj}
The main result of this section is \cref{thm:projective}.
\projective*

Observe that \cref{thm:projective}~{(a)} follows from \cref{thm:algo}, so we need to show the hardness counterpart, i.e., the statement {(b)}.
A crucial building block in our reduction will be the graph called the \emph{edge gadget}, whose construction is described in the following lemma.

\begin{lemma}\label{lem:projective-gadget}
For every non-trivial projective core $H$, there exists a graph $F$ with two specified vertices $u^*$ and $v^*$, satisfying the following:
\begin{compactenum}[(a)]
\item for every $x,y\in V(H)$ such that $x\neq y$, there exists a homomorphism $f:F \to H$ such that $f(u^*)=x$ and $f(v^*)=y$,
\item for every $f: F \to H$ it holds that $f(u^*)\neq f(v^*)$. 
\end{compactenum}
\end{lemma}
\begin{proof}
Let $V(H)=\{z_1,\dots,z_k\}$. For $i \in [k]$ denote by $z_i^{k-1}$ the $(k-1)$-tuple $(z_i, \dots, z_i)$ and by $\overline{z_i}$ the $(k-1)$-tuple $(z_1,\dots,z_{i-1},z_{i+1},\dots z_k)$. We claim that $F:=H^{(k-1)k}$ and vertices 
\[u^*:=(z_1^{k-1},\dots,z_k^{k-1}) \ \textrm{ and } \
v^*:=(\overline{z_1},\dots,\overline{z_k})\] satisfy the statement of the lemma. Note that the vertices of $F$ are $k(k-1)$-tuples.

To see that {(a)} holds, observe that if $x$ and $y$ are distinct vertices from $V(H)$, then there always exists $i\in[k(k-1)]$ such that $\pi_i(u^*)=x$ and $\pi_i(v^*)=y$. This means that $\pi_i$ is a homomorphism from $F=H^{k(k-1)}$ to $H$ satisfying $\pi_i(u^*)=x$ and $\pi_i(v^*)=y$.

To prove {(b)}, recall that since $H$ is projective, by \cref{obs:composition}, the homomorphism $f$ is a composition of some automorphism $g$ of $H$ and $\pi_i$ for some $i \in [k(k-1)]$. Observe that $u^*$ and $v^*$ are defined in a way such that $\pi_j(u^*)\neq \pi_j(v^*)$ for every $j \in [k(k-1)]$. As $g$ is an automorphism, it is injective, which gives us $f(u^*)=g(\pi_i(u^*))\neq g(\pi_i(v^*))=f(v^*)$.
\end{proof}

Finally, we are ready to prove \autoref{thm:projective} {(b)}.
The high-level idea is to start with an instance of \coloring{k}, where $k = |H|$, and replace each edge by the gadget constructed in~\cref{lem:projective-gadget}. Then the hardness will follow from \cref{thm:LMS}.
Actually, Lokshtanov, Marx, Saurabh~\cite{DBLP:journals/talg/LokshtanovMS18} proved the following, slightly stronger version of \cref{thm:LMS}.

\begin{cthm}{1'}[Lokshtanov, Marx, Saurabh~\cite{DBLP:journals/talg/LokshtanovMS18}] \label{thm:LMS-pw}
Let $k \geq 3$ be a fixed integer. Assuming the SETH, the \coloring{k} problem on a graph $G$ cannot be solved in time $\Oh^*\left((k-\epsilon)^{\pw{G}}\right)$ for any $\epsilon>0$. 
\end{cthm}

This allows us to prove a slightly stronger version of \autoref{thm:projective} {(b)}, where we consider the problem parameterized by the pathwidth of the instance graph.

\begin{cthm}{3' (b)} \label{thm:projective-lower-pw}
Let $H$ be a fixed non-trivial projective core on $k$ vertices. There is no algorithm solving the \homo{H} problem for instance graph $G$ in time $\Oh^*\left((k-\epsilon)^{\pw{G}}\right)$ for any $\epsilon >0$, unless the SETH fails.
\end{cthm}

\begin{proof}[Proof of \cref{thm:projective-lower-pw}]
Note that since $H$ is non-trivial, we have $k \geq 3$. Recall that since $H$ is projective, it is also connected~\cite{larose2001strongly}.
We reduce from \coloring{k}, let $G$ be an instance with $n$ vertices and pathwidth $t$. We construct an instance $G^*$ of \homo{H} as follows. First, for every $z \in V(G)$ we introduce a vertex $z'$ of $V(G^*)$. Let $V'$ denote the set of these vertices.
Now, for every edge $xy$ of $G$, we introduce to $G^*$ a copy of the edge gadget, constructed in \cref{lem:projective-gadget}, and denote it by $F_{xy}$. We identify the vertices $u^*$ and $v^*$ of $F_{xy}$ with vertices $x'$ and $y'$, respectively. This completes the construction of $G^*$.

We claim that $G$ is $k$-colorable if and only if $G^* \to H$.
Indeed, let $\phi$ be a $k$-coloring of $G$. For simplicity of notation, we label the colors used by $\phi$ in the same way as the vertices of $H$, i.e., $z_1,z_2,\dots, z_k$.
Define $g: V' \to V(H)$ by setting $g(v'):=\phi(v')$ Now consider an edge $xy$ of $G$ and the edge gadget $F_{xy}$.
Since $c$ is a proper coloring, we have $g(x')\neq g(y')$. So by \cref{lem:projective-gadget}~(a), we can find a homomorphism $f_{xy}:F_{xy} \to H$, such that $f_{xy}(x') = g(x')$ and $f_{xy}(y')=g(y')$. Repeating this for every edge gadget, we can extend $g$ to a homomorphism from $G^*$ to $H$.

Conversely, from \cref{lem:projective-gadget}~(b), we know that for any $f \colon G^* \to H$ and every edge $xy$ of $G$ it holds that $f(x') \neq f(y')$, so any homomorphism from $G^*$ to $H$ induces a $k$-coloring of $G$.

The number of vertices of $G^*$ is at most $|F|n^2$.
Now let $\cT$ be a path decomposition of $G$ of width $t$, denote its consecutive bags by $X_1,X_2,\ldots,X_m$. Let us extend it to a path decomposition of $G^*$. For each edge $xy$ of $G$ there exists $b \in [m]$ such that $x,y \in X_b$. We introduce a bag $X_{b'}:=X_b \cup V(F_{xy})$ as a direct successor of $X_b$. It is straightforward to observe that by repeating this step for every edge of $G$, we obtain a path decomposition of $G^*$ of width at most $t + |F|$. Recall that $H$ is fixed, so $|F|$ is a constant. So if we could decide if $G^*\to H$ in time $(k-\epsilon)^\pw{G^*}\cdot c \cdot |G^*|^d \leq (k-\epsilon)^{t + |F|} \cdot c \cdot |F|^d \cdot n^{2d}$, where $c$ and $d$ are some constants, then we would be able to decide if $G$ is $k$-colorable in time $(k-\epsilon)^t \cdot c' \cdot n^{d'}$ for constants $c'=c \cdot (k-\epsilon)^{|F|} \cdot |F|^d$ and $d' = 2d$. By \cref{thm:LMS-pw}, this contradicts the SETH.
\end{proof}

\subsection{Non-projective cores} \label{sec:nonproj}
Now we will focus on non-trivial connected cores, which are additionally non-projective, i.e., they do not satisfy the assumptions of \cref{thm:projective}. First, let us argue that the approach from \cref{sec:proj} cannot work in this case. In particular, we will show that an edge gadget with properties listed in \cref{lem:projective-gadget} cannot be constructed for non-projective graphs $H$.

We will need the definition of \emph{constructible sets}, see~Larose and Tardif~\cite{larose2001strongly}.
For a graph $H$, a set $C \subseteq V(H)$ is constructible if there exists a graph $K$, an $(\ell+1)$-tuple of vertices $x_0, \dots, x_\ell \in V(K)$ and an $\ell$-tuple of vertices $y_1, \dots, y_\ell \in V(H)$ such that
\[ \{y \in V(H) ~|~ \exists f: K \to H \textrm{ such that } f(x_i)=y_i \textrm{ for every }i\in[\ell] \textrm{ and } f(x_0)=y \}=C.\]
We can think of $C$ as the set of colors that might appear on the vertex $x_0$, when we precolor each $x_i$ with the color $y_i$ and try to extend this partial mapping to a homomorphism to $H$. The tuple $(K,x_0,\ldots,x_\ell,y_1,\ldots,y_\ell)$ is called a \emph{construction} of $C$.

It appears that the notion of constructible sets is closely related to projectivity.
\begin{theorem}[Larose, Tardif~\cite{larose2001strongly}]
A graph $H$ on at least three vertices is projective if and only if every subset of its vertices is constructible.
\end{theorem}

Now we show that \cref{lem:projective-gadget} cannot work for non-projective graphs $H$.

\begin{proposition}
Let $H$ be a fixed non-trivial connected core.
Then an edge gadget $F$ with properties listed in \cref{lem:projective-gadget}  exists if and only if $H$ is projective.
\end{proposition}
\begin{proof}
The `if' statement follows from \cref{lem:projective-gadget}. Let $k:=|H|$ and suppose that there exists a graph $F$ with properties given in \cref{lem:projective-gadget}.
Consider a set $C \subseteq V(H)$ and define $\ell := |C|$. Let $\{y_1, \dots, y_{k-\ell}\}$ be the complement of $C$ in $V(H)$. Take $k-\ell$ copies of $F$, say $F_1, \dots ,F_{k-\ell}$ and denote the vertices $u^*$ and $v^*$ of the $i$-th copy $F_i$  by $u^*_i$ and $v^*_i$, respectively. Identify the vertices $u^*_i$ of all these copies, denote the obtained vertex by $u^*$, and the obtained graph by $K$.
Now set $x_0 := u^*$ and for each $i \in [k-\ell]$ set $x_i:=v^*_i$.

It is easy to verify that this is a construction of the set~$C$. Indeed, observe that if $x \in C$, then, from \cref{lem:projective-gadget}~(a), for each copy $F_i$ there exists a homomorphism $f_i:F_i \to H$ such that $f_i(v^*_i)=f_i(x_i)=y_i$ and $f_i(u^*)=f_i(x_0)=x$. Combining these homomorphisms yields a homomorphism $f: K \to H$. On the other hand, if $x \not\in C$, then $x=y_i$ for some $i \in [k-\ell]$. But from \cref{lem:projective-gadget}~(b) we know that for every homomorphism $f \colon F_i \to H$ it holds that $x=y_i=f(v^*_i)\neq f(u^*)=f(x_0)$, so $x_0$ cannot be mapped to $x$ by any extension to a homomorphism from $K$ to $H$. 
\end{proof}

Observe that if $H$ is projective, then it must be indecomposable. Indeed, assume that for some non-trivial $H$ it holds that $H = H_1 \times H_2$, $H \not\simeq K_1^*$ and $H_2 \not\simeq K_1^*$. Consider a homomorphism $f: (H_1 \times H_2)^2 \to H_1 \times H_2$, defined as $f((x,y),(x',y'))=(x,y')$. Note that it is idempotent, but not a projection, so $H$ is not projective. 

In the light of the observation above, it is natural to ask whether indecomposability implies projectivity.
This problem was already stated e.g. by Larose and Tardif~\cite[Problem 2]{larose2001strongly} and, to the best of our knowledge, no significant progress in this direction was made. Let us recall it here.


\conprojprime*

Since we know no connected non-trivial non-projective cores that are indecomposable, in the remainder of the section we will assume that $H$ is a decomposable, non-trivial connected core. By \cref{thm:unique-fact} we know that $H$ has a unique prime factorization $H_1 \times \ldots \times H_m$ for some $m \geq 2$.
To simplify the notation, for any given homomorphism $f:G \to H_1 \times \ldots \times H_m$ and $i \in [m]$, we define $f_i \equiv\pi_i \circ f$. Then for each vertex $x$ of $G$ it holds that \[f(x)=(f_1(x),\ldots,f_m(x)),\]
and $f_i$ is a homomorphism from $G$ to $H_i$.

The following observation follows from \cref{obs:projprop}.
\begin{observation}\label{obs:necessary-to-core}
Let $H$ be a connected, non-trivial core with factorization $H=H_1 \times  \ldots \times H_m$, such that $H_i \not\simeq K_1^*$ for all $i \in [m]$.
Then for $i \in [m]$ the graph $H_i$ is a connected non-trivial core, incomparable with $H_j$ for $j \in [m] \setminus \{i\}$.{\hfill$\square$\smallskip}
\end{observation}

Now let us consider the complexity of \homo{H}, where $H$ has a prime factorization $H_1 \times H_2 \times \ldots H_m$ for $m \geq 2$. By \cref{thm:non-prime}, the problem can be solved in time $\Oh\left(|H|^4+\max_{j \in [m]}|H_j|^{t+1}\cdot n\right)$, where $n$ and $t$ are, respectively, the number of vertices and the width of a given tree decomposition of the input graph.
We believe that this bound is tight, and we prove a matching lower bound, up to the polynomial factor, under some additional assumption.

We say that a graph $H$ is \emph{truly projective} if it has at least three vertices and for every $s \geq 2$ and every connected core $W$ incomparable with $H$, it holds that if $f : H^s \times W \to H$ satisfies $f(x,x,\dots,x,y)=x$ for all $x \in V(H), y \in V(W)$, then $f$ is a projection.

It is easy to verify that truly projective graphs are projective. Indeed, by  \cref{thm:proj-char}, we need to show that any idempotent homomorphism $g \colon H^2 \to H$ is a projection. Consider a core $W$, which is incomparable with $H$,  and a homomorphism $f \colon H^2 \times W \to H$, defined by $f(x_1,x_2,y) := g(x_1,x_2)$. Since $H$ is truly projective, $f$ is a projection, and so is $g$.

Again, we state and prove the stronger version of a lower bound, parameterized by the pathwidth of an instance graph.

\begin{theorem}\label{thm:nonprojective-lower-h-pw}
Let $H$ be a fixed non-trivial connected core, with prime factorization $H_1 \times \ldots \times H_m$. 
Assume there exists $i \in [m]$ such that $H_i$ is truly projective.
Unless the SETH fails, there is no algorithm solving the \homo{H} problem for instance graph $G$ in time $\Oh^*\left((|H_i|-\epsilon)^\pw{G}\right)$, for any $\epsilon >0$.
\end{theorem}


The proof of \cref{thm:nonprojective-lower-h-pw} is similar to the proof of \cref{thm:projective-lower-pw}. We start with constructing an appropriate edge gadget. We will use the following result (to avoid introducing new definitions, we state the theorem in a sightly weaker form, using the terminology used in this paper, see also~\cite[Theorem 8.18]{hammack2011handbook}).

\begin{theorem}[D{\"o}rfler, \cite{dorfler1974primfaktorzerlegung}]\label{thm:auth}
Let $\varphi$ be an automorphism of a connected, non-bipartite, ramified graph $H$, with the prime factorization $H_1 \times \ldots \times H_m$. Then for each $i \in [m]$ there exists an automorphism $\varphi^{(i)}$ of $H_i$ such that $\varphi_i(t_1,\dots, t_m)\equiv \varphi^{(i)}(t_i).$
\end{theorem}

In particular, it implies the following.
\begin{corollary}\label{col:auto}
Let $\mu$ be an automorphism of a connected, non-trivial core $H =H_1 \times R$, where $H_1$ is indecomposable and $R \not\simeq K_1^*$.
Then there exist automorphisms $\mu^{(1)}:H_1 \to H_1$ and $\mu^{(2)}:R \to R$ such that $\mu(t,t')\equiv(\mu^{(1)}(t),\mu^{(2)}(t'))$.
\end{corollary} 
\begin{proof}
By \cref{obs:necessary-to-core}, $R$ is a non-trivial core, so it admits the unique prime factorization, say $R=H_2 \times \ldots \times H_m$. Therefore $H_1 \times H_2 \times \ldots \times H_m$ is the unique prime factorization of $H$.
From \cref{thm:auth} we know that for each $i \in [m]$ there exists an automorphism $\varphi^{(i)}$ of $H_i$ such that $\mu(t_1,\dots, t_m)\equiv (\varphi^{(1)}(t_1),\dots,\varphi^{(m)}(t_m))$.
Define $\mu^{(1)}$ by setting $\mu^{(1)}(t):=\varphi^{(1)}(t)$ for every vertex $t \in V(H_1)$.
Analogously, we define $\mu^{(2)}$ by setting $\mu^{(2)}(t_2,\ldots,t_m):=(\varphi^{(2)}(t_2), \ldots, \varphi^{(m)}(t_m))$ for every vertex $(t_2,\ldots,t_m)$ of $R$ (for each $i \in [m] \setminus \{1\}$ we have $t_i \in V(H_i)$).
It is straightforward to verify that $\mu^{(1)}$ and $\mu^{(2)}$ satisfy the statement of the corollary.
\end{proof}


In the following lemma we construct an edge gadget, that will be used in the hardness reduction. The construction is similar to the one in \cref{lem:projective-gadget}, but more technically complicated.

\begin{lemma}\label{lem:nonprojective-gadget}
Let $H=H_1 \times R$ be a connected, non-trivial core, such that $H_1$ is truly projective and $R \not\simeq K_1^*$. Let $w$ be a fixed vertex of $R$. Then there exists a graph $F$ and vertices $u^*, v^*$ of $F$, satisfying the following conditions:
\begin{compactenum}[(a)]
\item for every $xy \in E(H_1)$ there exists $f: F \to H$ such that $f(u^*)=(x,w)$ and $f(v^*)=(y,w)$,
\item for any $f: F \to H$ it holds that $f_1(u^*)f_1(v^*) \in E(H_1)$.
\end{compactenum}
\end{lemma}
\begin{proof}
Let $E(H_1)=\{e_1,\ldots, e_s\}$ and let $e_i=u_iv_i$ for every $i \in [s]$ (clearly, one vertex can appear many times as some $u_i$ or $v_j$). Consider the vertices 
\begin{align*}
u:= &(u_1, \dots, u_s, v_1, \dots, v_s)\\
v:= &(v_1, \dots, v_s, u_1, \dots, u_s)
\end{align*} of $H_1^{2s}$. Let $F:=H_1^{2s} \times R$, and let $u^*:=(u,w)$ and $v^*:=(v,w)$. We will treat vertices $u$ and $v$ as $2s$-tuples, and vertices $u^*$ and $v^*$ as $(2s+1)$-tuples.

Observe that, if $xy \in E(H_1)$, then, by the definition of $u^*$ and $v^*$, there exists $i \in[2s]$ such that $x=\pi_i(u)$ and $y=\pi_i(v)$. Define a function $f: V(F) \to V(H)$ as $f(x_1,\ldots,x_{2s},r): = (\pi_i(x_1,\dots,x_{2s}), r)$. Observe that this is a homomorphism, for which $f(u^*)=f(u,w)=(x,w)$ and $f(v^*)=f(v,w)=(y,w)$, which is exactly the condition (a) in the statement of \cref{lem:nonprojective-gadget}. 

We prove (b) in two steps. First, we observe the following.
\begin{claim*}\label{cla:gadget-oranother}
Let $\varphi: F \to H$. If for every $z \in V(H_1)$ and $r \in V(R)$ it holds that $\varphi_1(z,\dots,z,r)=z$ then $\varphi_1(u^*)\varphi_1(v^*)\in E(H_1)$.
\end{claim*}
\begin{inproof}
Recall that $R$ is a connected core incomparable with $H_1$, and $H_1$ is truly projective. It means that if $\varphi_1:H_1^{2s} \times R \to H_1$ satisfies the assumption of the claim, then it is equal to $\pi_i$ for some $i \in [2s]$. From the definition of $u^*$ and $v^*$ we have that $\pi_i(u^*)\pi_i(v^*)\in E(H_1)$.
\end{inproof}

Note that the set $\{(z,\dots,z, r) \in F~|~ z \in V(H_1), r \in V(R)\}$ induces in $F$ a subgraph isomorphic to $H$, let us call it $\widetilde{H}$. Let $\sigma$ be an isomorphism from $\widetilde{H}$ to $H$ defined as $\sigma(z,\dots,z,r):=(z,r)$. 

Consider any homomorphism $f \colon F \to H$. We observe that $f|_{V(\widetilde{H})}$ is an isomorphism from $\widetilde{H}$ to $H$, because $H$ is a core. If $f |_{V(\widetilde{H})}\equiv \sigma$ then for every $z \in V(H)$ and $r \in V(R)$ it holds that $f_1(z,\dots,z,r)= \sigma_1(z,\dots,z,r)=z$, so, by the Claim above, we are done.
If not, observe that there exists the inverse isomorphism $g:H \to \widetilde{H}$ such that $g \circ f |_{V(\widetilde{H})}$ is the identity function on $V(\widetilde{H})$.
Define $\mu:=\sigma \circ g$. Observe that $\mu$ is an endomorphism of $H_1 \times R$, so an automorphism, since $H_1 \times R$ is a core. Also note that $(\mu \circ f): F \to H_1 \times R$ is a homomorphism such that for every $(z,\ldots,z,r) \in V(\widetilde{H})$ it holds that
\[(\mu \circ f)(z,\dots,z,r)=(\sigma \circ g \circ f)(z,\dots,z,r)=(\sigma \circ \id)(z,\dots,z,r)=\sigma(z,\dots,z,r) = (z,r),\]
so $(\mu \circ f)_1(z,\dots,z,z')=z$. This means that $\mu \circ f$ satisfies the assumption of the Claim, so 
\begin{align}\label{eq:edge}
(\mu \circ f)_1(u^*)(\mu \circ f)_1(v^*) \in E(H_1).
\end{align}
Clearly, for every vertex $\bar{z}$ of $F$ it holds that
\begin{align}\label{eq:mu}
\begin{split}
\left(\mu \circ f\right)(\bar{z}) =\mu\Big(f_1(\bar{z}), f_2(\bar{z})\Big)
= \Bigg( \mu_1\Big(f_1(\bar{z}), f_2(\bar{z})\Big),\mu_2\Big(f_1(\bar{z}), f_2(\bar{z})\Big)\Bigg).
\end{split}
\end{align}
Note that \cref{col:auto} implies that there exist automorphisms $\mu^{(1)}$ and $\mu^{(2)}$ of $H_1$ and $R$, respectively, such that for every $\bar{z} \in V(F)$ it holds that
\begin{align}\label{eq:autom}
\begin{split}
\mu_1\big(f_1(\bar{z}), f_2(\bar{z})\big) = & \mu^{(1)} (f_1(\bar{z}))\\
\mu_2\big(f_1(\bar{z}), f_2(\bar{z})\big) = & \mu^{(2)} (f_2(\bar{z})),
\end{split}
\end{align}
In particular, \eqref{eq:mu} and \eqref{eq:autom} imply that $\left(\mu \circ f\right)_1 = \mu^{(1)} \circ f_1$. Combining this with \eqref{eq:edge} we get that
\begin{align}\label{eq:edge2}
\left(\mu^{(1)} \circ f_1\right)(u^*)\left(\mu^{(1)} \circ f_1\right)(v^*) \in E(H_1).
\end{align}
Since $\mu^{(1)}$ is the automorphism of $H_1$, there exists the inverse automorphism $\left(\mu^{(1)}\right)^{-1}$ of $H_1$. Because $\left(\mu^{(1)}\right)^{-1}$ is an automorphism, \eqref{eq:edge2} implies that  $f_1(u^*)f_1(v^*) \in E(H_1)$, which completes the proof.
\end{proof}

Now we can proceed to the proof of \cref{thm:nonprojective-lower-h-pw}.
\begin{proof}[Proof of \autoref{thm:nonprojective-lower-h-pw}]
Since $\times$ is commutative, without loss of generality we can assume that $H_1$ is truly projective. Define $R:=H_2 \times \ldots \times H_m$, so $H = H_1 \times R$. Since $H_1$ is truly projective, it is projective, so \cref{thm:projective-lower-pw} can be applied here. Hence we known that assuming the SETH, there is no algorithm which solves instances of \homo{H_1} with $n$ vertices and pathwidth $t$ in time $\Oh^*\left((|H_1|-\epsilon)^t\right)$, for any $\epsilon >0$.

Let $G$ be an instance of \homo{H_1} with $n$ vertices and pathwidth $t$.
The construction of the instance $G^*$ of \homo{H} is analogous as in the proof of \cref{thm:projective-lower-pw}.
Let $w$ be a fixed vertex of $R$ and let $F$ be a graph obtained by calling \cref{lem:nonprojective-gadget} for $H$ and $w$.
For every vertex $z$ of $G$, we introduce to $G^*$ a vertex $z'$. Then we add a copy $F_{xy}$ of $F$ for every pair of vertices $x',y'$, which corresponds to an edge $xy$ in $G$, and identify vertices $x'$ and $y'$ with vertices $u^*$ and $v^*$ of $F_{xy}$, respectively.
As in the proof of \cref{thm:projective-lower-pw}, we observe that $G^*$ is a yes-instance of \homo{H} if and only if $G$ is a yes-instance of \homo{H_1}. Moreover,  $|G^*| \leq |F| \cdot n^2$ and $\pw{G^*} \leq t + |F|$. Thus, if we could decide if $G^*\to H$ in time $\Oh^*\left((|H_1|-\epsilon)^{\pw{G^*}}\right)$, then we would be able to decide if $G \to H_1$ in time $\Oh^* \left( (|H_1|-\epsilon)^t \right)$. By \cref{thm:projective-lower-pw}, such an algorithm contradicts the SETH.
\end{proof}

Note that combining the results from \cref{thm:non-prime} and \cref{thm:nonprojective-lower-h-pw} we obtain a tight complexity bound for graphs $H$, whose largest factor is truly projective.

\begin{corollary}\label{col:non-projective}
Let $H$ be a non-trivial, connected core with prime factorization $H_1 \times \ldots \times H_m$ and let $H_i$ be the factor with the largest number of the vertices. Assume that $H_i$ is truly projective. Let $n$ and $t$ be, respectively, the number of vertices and the treewidth of an instance graph $G$.
\begin{compactenum}[(a)]
\item Even if $H$ is a part of the input, the \homo{H} problem can be solved in time $\Oh\left(|H|^4 + |H_i|^{t+1}\cdot n\right)$, if a tree decomposition of $G$ of width $t$ is given.
\item Even if $H$ is fixed, there is no algorithm solving the \homo{H} problem in time $\Oh^*\left((|H_i|-\epsilon)^t\right)$ for any $\epsilon >0$, unless the SETH fails. {\hfill$\square$\smallskip}
\end{compactenum}
\end{corollary}

\section{Conclusion} \label{sec:conclusion}

Recall that in \cref{thm:nonprojective-lower-h-pw} and \cref{col:non-projective} we presented lower complexity bounds for \homo{H} in the case that one of factors of $H$ is truly projective. In the light of \cref{con:strongly}, we would like to weaken this assumption by substituting ``truly projective'' with ``projective''. Let us discuss the possibility of obtaining such a result.

As mentioned in the introduction, a class of graphs very close to truly projective graphs was considered by Larose~\cite{Larose2002FamiliesOS}. In the same paper, he defined and studied the so-called \emph{strongly projective graphs}. A graph $H$ on at least three vertices is strongly projective, if for every connected graph $W$ on at least two vertices and every $s \geq 2$, the only homomorphisms $f: H^s \times W  \to H$ satisfying $f(x, \dots, x, y) = x$ for all $x \in V(H)$ and $y \in V(W)$, are projections.
Note that this definition is very similar, but more restrictive than the definition of truly projective graphs. Indeed, for truly projective graphs $H$ we restricted the homomorphisms from $H^s \times W$ to $H$ only for connected cores $W$, that are incomparable with $H$.
Thus it is clear that every strongly projective graph is truly projective, and, as observed before, every truly projective graph is projective.
Among other properties of strongly projective graphs, Larose~\cite{Larose2002FamiliesOS,larose2002strongly} shows that their recognition is decidable -- note that this does not follow directly from the definition.

Let us recall some results on strongly projective graphs, as they show that many natural graphs satisfy the assumptions of \cref{thm:nonprojective-lower-h-pw} and \cref{col:non-projective}. 
We say that graph is \emph{square-free} if it does not contain a copy of $C_4$ as a (not necessarily induced) subgraph. Larose proved the following.

\begin{theorem}[Larose~\cite{Larose2002FamiliesOS}]\label{thm:squarefreee}
If $H$ is a square-free, connected, non-bipartite core, then it is strongly projective.
\end{theorem}

\begin{example*}Consider the graph $G_B$ on 21 vertices, shown on~\cref{fig:brinkmann} (left), it is called \emph{the Brinkmann graph}~\cite{brinkmann1997smallest}. It is connected, its chromatic number is 4 and its girth is 5. In particular, it is square-free. Thus by \cref{thm:squarefreee} we know that $G_B$ is strongly projective. 
By exhaustive computer search we verified that $K_3 \times G_B$ is a core. Let us consider the complexity of \homo{K_3 \times G_B} for input graphs with $n$ vertices and treewidth $t$.
The straightforward dynamic programming approach from \cref{thm:algo} results in the running time $\Ohs(63^t)$. However, \cref{thm:non-prime} gives us a faster algorithm, whose running time is $\Ohs(21^t)$. Moreover, by~\cref{col:non-projective} we know that this algorithm is likely to be asymptotically optimal, i.e., there is no algorithm with running time $\Ohs((21-\epsilon)^t)$ for any $\epsilon >0$ and any constants $c,d$, unless the SETH fails.
\end{example*}

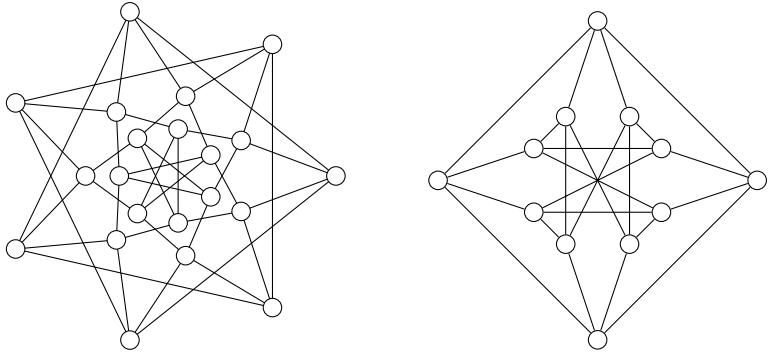
\begin{figure}[ht]
\centering{
\begin{tikzpicture}[scale=0.8,every node/.style={draw,circle,fill=white,inner sep=0pt,minimum size=7pt},every loop/.style={},scale=0.8]
\def\n{7}
\foreach \i in {1,...,\n}
{
	\draw (360/\n*\i-720/\n:3.5) -- (360/\n*\i:3.5);
	\draw (360/\n*\i+180:1.7) -- (360/\n*\i-1080/\n:3.5);
	\draw (360/\n*\i+180:1.7) -- (360/\n*\i+1080/\n:3.5);
	\draw (360/\n*\i+180:1.7) -- (360/\n*\i+180+360/\n:1);
	\draw (360/\n*\i+180:1.7) -- (360/\n*\i+180-360/\n:1);
	\draw (360/\n*\i+180+720/\n:1) -- (360/\n*\i+180-360/\n:1);
}

\foreach \i in {1,...,\n}
{
	\node (a\i) at (360/\n*\i:3.5) {};	
	\node (a\i) at (360/\n*\i+180:1.7) {};
	\node (a\i) at (360/\n*\i+180:1) {};
}
\end{tikzpicture} \hskip 1cm
\begin{tikzpicture}[scale=0.75,every node/.style={draw,circle,fill=white,inner sep=0pt,minimum size=7pt},every loop/.style={}, rotate=45,scale=0.8]
\node (a) at (0,0) {};
\node (b) at (5,0) {};
\node (d) at (0,5) {};
\node (c) at (5,5) {};

\node (e) at (1,2) {};
\node (f) at (2,1) {};
\node (j) at (3,4) {};
\node (i) at (4,3) {};
\node (g) at (3,1) {};
\node (h) at (4,2) {};
\node (l) at (1,3) {};
\node (k) at (2,4) {};

\draw (a) -- (b) -- (c) -- (d) -- (a);
\draw (a) -- (e) -- (i) -- (h) -- (l) -- (e) -- (j);
\draw (a) -- (f) -- (g) -- (k) -- (j) -- (f) -- (i);
\draw (g) -- (b) -- (h) -- (k) -- (d) -- (l) -- (g);
\draw (i) -- (c) -- (j);
\end{tikzpicture}
}
\caption{The Brinkmann graph (left) and the Chv\'{a}tal graph (right).}
\label{fig:brinkmann}
\end{figure}

A graph is said to be \emph{primitive} if there is no non-trivial partition of its vertices which is invariant under all automorphisms of this graph (see e.g.~\cite{smith1971primitive}).

\begin{theorem}[Larose~\cite{Larose2002FamiliesOS}]\label{thm:primitive}
If $H$ is an indecomposable primitive core, then it is strongly projective.
\end{theorem}

In particular, \cref{thm:primitive} implies that Kneser graphs are strongly projective \cite{Larose2002FamiliesOS}. Note that Kneser graphs might have 4-cycles, so this statement does not follow from~\cref{thm:squarefreee}.

Interestingly, Larose~\cite{Larose2002FamiliesOS,larose2002strongly} proved that members of all known families of projective graphs are in fact strongly projective (and thus of course truly projective). He also asked whether the same holds for all projective graphs. We recall this problem in a weaker form, which would be sufficient in our setting.

\conjstrongly*

Clearly, if both \cref{con:proj-iff-prime} and \cref{con:strongly} are true, there is another characterization of non-trivial indecomposable connected cores.

\begin{observation}\label{obs:star-iff-prime}
Assume that \cref{con:proj-iff-prime} and \cref{con:strongly} hold. Let $H$ be a connected non-bipartite core. Then $H$ is indecomposable if and only if it is truly projective. {\hfill$\square$\smallskip}
\end{observation}

Note that \cref{thm:projective}, \cref{col:non-projective}, and \cref{obs:star-iff-prime} imply the following result.

\thmeverything*

We believe that the bounds from \cref{thm:everything} are tight for all connected cores. Recall that in \cref{obs:disconnected-not-fixed} and \cref{lem:disconnected-fixed} we showed that in order to prove tight bounds for the \homo{H} problem, we can restrict ourselves to connected cores $H$. Observe that combining these reductions with the result of \cref{thm:everything}, we obtain the following complexity bounds.

\begin{theorem}\label{thm:everything+discon}
Assume that \cref{con:proj-iff-prime} and \cref{con:strongly} hold.
Let $H=H_1+ \ldots + H_\ell$ be a non-trivial core and let $H_{i,1 }\times \ldots \times H_{i,m_i}$ be the prime factorization of $H_i$, for every $i \in [\ell]$.
Define $k:= \max_{i \in [\ell], j \in [m_i]} |H_{i,j}|$. Let $n$ and $t$ be, respectively, the number of vertices and the treewidth of an instance graph $G$.
\begin{compactenum}[(a)]
\item Even if $H$ is given as an input, the \homo{H} problem can be solved in time $\Oh(|H|^4+n \cdot k^{t+1}\cdot |H|)$, assuming a tree decomposition of $G$ of width $t$ is given.
\item Even if $H$ is fixed, there is no algorithm solving the \homo{H} problem in time $\Oh^*\left((k-\epsilon)^t\right)$ for any $\epsilon >0$, unless the SETH fails.
\end{compactenum}
\end{theorem}


Finally, let us point out one more problem, related to the ones discussed in this paper. Recall that if $H=H_1 \times H_2$ is a connected, non-trivial core and $H_1 \not\simeq K_1^*, H_2 \not\simeq K_1^*$, then $H_1$ and $H_2$ must be incomparable cores. We believe it would be interesting to know if the opposite implication holds as well. To motivate the study on this problem, we state the following conjecture.

\begin{conjecture}\label{con:product-core}
Let $H_1$ and $H_2$ be connected, indecomposable, incomparable cores. Then $H_1 \times H_2$ is a core.\end{conjecture}

Note that it is straightforward to verify that if $H_1$ and $H_2$ are ramified, then their direct product $H_1 \times H_2$ is ramified as well. As every core is in a particular ramified, this is a necessary condition for \cref{con:product-core} to hold.

We confirmed the conjecture by exhaustive computer search for some small graphs.
In particular, the conjecture is true for graphs $K_3 \times H$, where $H$ is any 4-vertex-critical, triangle-free graph with at most 14 vertices~\cite{DBLP:journals/dam/BrinkmannCGM13}, the Gr\"otzsch graph (see ~\cref{fig:grotzsch}), the Brinkmann graph (see~\cref{fig:brinkmann} (left)), or the Chv\'{a}tal graph (see~\cref{fig:brinkmann} (right)). 

Let us point out that the spirit of \cref{con:product-core} is similar to the spirit of the recently disproved \emph{Hedetniemi's conjecture}~\cite{hedetniemi1966homomorphisms, sauer2001hedetniemi,shitov2019counterexamples,tardif2008hedetniemi}, which also asked how the properties of homomorphisms of factor graphs affect the properties of homomorphisms of their product.

\paragraph*{Acknowledgment.} The authors are grateful to D. Marx for introducing us to the problem, and to B. Larose, C. Tardif, B. Martin, and Mi. Pilipczuk for useful comments.


\begin{thebibliography}{10}

\bibitem{Arnborg1987}
S.~Arnborg, D.~G. Corneil, and A.~Proskurowski.
\newblock Complexity of finding embeddings in a $k$-tree.
\newblock {\em {SIAM} Journal on Algebraic Discrete Methods}, 8(2):277--284,
  Apr. 1987.

\bibitem{DBLP:journals/dam/ArnborgP89}
S.~Arnborg and A.~Proskurowski.
\newblock Linear time algorithms for {NP}-hard problems restricted to partial
  $k$-trees.
\newblock {\em Discrete Applied Mathematics}, 23(1):11--24, 1989.

\bibitem{bertele1972nonserial}
U.~Bertele and F.~Brioschi.
\newblock {\em Nonserial dynamic programming}.
\newblock Academic Press, 1972.

\bibitem{DBLP:journals/siamcomp/BjorklundHK09}
A.~Bj{\"{o}}rklund, T.~Husfeldt, and M.~Koivisto.
\newblock Set partitioning via inclusion-exclusion.
\newblock {\em {SIAM} J. Comput.}, 39(2):546--563, 2009.

\bibitem{bodlaender2013fine}
H.~L. Bodlaender, P.~Bonsma, and D.~Lokshtanov.
\newblock The fine details of fast dynamic programming over tree
  decompositions.
\newblock In {\em International Symposium on Parameterized and Exact
  Computation}, pages 41--53. Springer, 2013.

\bibitem{DBLP:journals/iandc/BodlaenderCKN15}
H.~L. Bodlaender, M.~Cygan, S.~Kratsch, and J.~Nederlof.
\newblock Deterministic single exponential time algorithms for connectivity
  problems parameterized by treewidth.
\newblock {\em Inf. Comput.}, 243:86--111, 2015.

\bibitem{DBLP:journals/cj/BodlaenderK08}
H.~L. Bodlaender and A.~M. C.~A. Koster.
\newblock Combinatorial optimization on graphs of bounded treewidth.
\newblock {\em Comput. J.}, 51(3):255--269, 2008.

\bibitem{DBLP:journals/dam/BrinkmannCGM13}
G.~Brinkmann, K.~Coolsaet, J.~Goedgebeur, and H.~M{\'{e}}lot.
\newblock {House of Graphs: A} database of interesting graphs.
\newblock {\em Discrete Applied Mathematics}, 161(1-2):311--314, 2013.

\bibitem{brinkmann1997smallest}
G.~Brinkmann and M.~Meringer.
\newblock The smallest 4-regular 4-chromatic graphs with girth 5.
\newblock {\em Graph Theory Notes of New York}, 32:40--41, 1997.

\bibitem{DBLP:journals/tcs/Bulatov05}
A.~A. Bulatov.
\newblock {$H$}-coloring dichotomy revisited.
\newblock {\em Theor. Comput. Sci.}, 349(1):31--39, 2005.

\bibitem{DBLP:journals/iandc/Courcelle90}
B.~Courcelle.
\newblock The monadic second-order logic of graphs. {I. Recognizable} sets of
  finite graphs.
\newblock {\em Inf. Comput.}, 85(1):12--75, 1990.

\bibitem{DBLP:journals/jacm/CyganFGKMPS17}
M.~Cygan, F.~V. Fomin, A.~Golovnev, A.~S. Kulikov, I.~Mihajlin, J.~Pachocki,
  and A.~Soca{\l}a.
\newblock Tight lower bounds on graph embedding problems.
\newblock {\em J. {ACM}}, 64(3):18:1--18:22, 2017.

\bibitem{cygan2015parameterized}
M.~Cygan, F.~V. Fomin, {\L}.~Kowalik, D.~Lokshtanov, D.~Marx, M.~Pilipczuk,
  M.~Pilipczuk, and S.~Saurabh.
\newblock {\em Parameterized algorithms}.
\newblock Springer, 2015.

\bibitem{dorfler1974primfaktorzerlegung}
W.~D{\"o}rfler.
\newblock {Primfaktorzerlegung und Automorphismen des Kardinalproduktes von
  Graphen}.
\newblock {\em Glasnik Matemati\v{c}ki}, 9:15--27, 1974.

\bibitem{DBLP:conf/stacs/EgriMR18}
L.~Egri, D.~Marx, and P.~Rz{\k{a}}{\.{z}}ewski.
\newblock Finding list homomorphisms from bounded-treewidth graphs to reflexive
  graphs: a complete complexity characterization.
\newblock In R.~Niedermeier and B.~Vall{\'{e}}e, editors, {\em 35th Symposium
  on Theoretical Aspects of Computer Science, {STACS} 2018, February 28 to
  March 3, 2018, Caen, France}, volume~96 of {\em LIPIcs}, pages 27:1--27:15.
  Schloss Dagstuhl - Leibniz-Zentrum fuer Informatik, 2018.

\bibitem{erdos1959graph}
P.~Erd\H{o}s.
\newblock Graph theory and probability.
\newblock {\em Canadian Journal of Mathematics}, 11:34--38, 1959.

\bibitem{FEDER1998236}
T.~Feder and P.~Hell.
\newblock List homomorphisms to reflexive graphs.
\newblock {\em Journal of Combinatorial Theory, Series B}, 72(2):236 -- 250,
  1998.

\bibitem{DBLP:journals/combinatorica/FederHH99}
T.~Feder, P.~Hell, and J.~Huang.
\newblock List homomorphisms and circular arc graphs.
\newblock {\em Combinatorica}, 19(4):487--505, 1999.

\bibitem{DBLP:journals/jgt/FederHH03}
T.~Feder, P.~Hell, and J.~Huang.
\newblock Bi-arc graphs and the complexity of list homomorphisms.
\newblock {\em Journal of Graph Theory}, 42(1):61--80, 2003.

\bibitem{DBLP:journals/mst/FominHK07}
F.~V. Fomin, P.~Heggernes, and D.~Kratsch.
\newblock Exact algorithms for graph homomorphisms.
\newblock {\em Theory Comput. Syst.}, 41(2):381--393, 2007.

\bibitem{greenwell1974applications}
D.~Greenwell and L.~Lov{\'a}sz.
\newblock Applications of product colouring.
\newblock {\em Acta Mathematica Hungarica}, 25(3-4):335--340, 1974.

\bibitem{Halin1976}
R.~Halin.
\newblock S-functions for graphs.
\newblock {\em Journal of Geometry}, 8(1):171--186, Mar 1976.

\bibitem{hammack2011handbook}
R.~Hammack, W.~Imrich, and S.~Klav{\v{z}}ar.
\newblock {\em Handbook of product graphs}.
\newblock CRC press, 2011.

\bibitem{hammack2009cartesian}
R.~H. Hammack and W.~Imrich.
\newblock On cartesian skeletons of graphs.
\newblock {\em Ars Mathematica Contemporanea}, 2(2):191--205, 2009.

\bibitem{hedetniemi1966homomorphisms}
S.~T. Hedetniemi.
\newblock Homomorphisms of graphs and automata.
\newblock Technical report, 1966.

\bibitem{DBLP:journals/jct/HellN90}
P.~Hell and J.~Ne{\v{s}}et{\v{r}}il.
\newblock On the complexity of \emph{H}-coloring.
\newblock {\em J. Comb. Theory, Ser. {B}}, 48(1):92--110, 1990.

\bibitem{DBLP:journals/dm/HellN92}
P.~Hell and J.~Ne{\v{s}}et{\v{r}}il.
\newblock The core of a graph.
\newblock {\em Discrete Mathematics}, 109(1-3):117--126, 1992.

\bibitem{hell2004graphs}
P.~Hell and J.~Ne{\v{s}}et{\v{r}}il.
\newblock {\em Graphs and homomorphisms}.
\newblock Oxford University Press, 2004.

\bibitem{DBLP:journals/jcss/ImpagliazzoP01}
R.~Impagliazzo and R.~Paturi.
\newblock On the complexity of $k$-{SAT}.
\newblock {\em J. Comput. Syst. Sci.}, 62(2):367--375, 2001.

\bibitem{DBLP:journals/jcss/ImpagliazzoPZ01}
R.~Impagliazzo, R.~Paturi, and F.~Zane.
\newblock Which problems have strongly exponential complexity?
\newblock {\em J. Comput. Syst. Sci.}, 63(4):512--530, 2001.

\bibitem{DBLP:journals/corr/KociumakaP17}
T.~Kociumaka and M.~Pilipczuk.
\newblock Deleting vertices to graphs of bounded genus.
\newblock {\em Algorithmica}, 81(9):3655--3691, 2019.

\bibitem{Larose2002FamiliesOS}
B.~Larose.
\newblock Families of strongly projective graphs.
\newblock {\em Discussiones Mathematicae Graph Theory}, 22:271--292, 2002.

\bibitem{larose2002strongly}
B.~Larose.
\newblock Strongly projective graphs.
\newblock {\em Canadian Journal of Mathematics}, 54(4):757--768, 2002.

\bibitem{larose2001strongly}
B.~Larose and C.~Tardif.
\newblock Strongly rigid graphs and projectivity.
\newblock {\em Multiple-Valued Logic}, 7:339--361, 2001.

\bibitem{DBLP:journals/eatcs/LokshtanovMS11}
D.~Lokshtanov, D.~Marx, and S.~Saurabh.
\newblock Lower bounds based on the {Exponential Time Hypothesis}.
\newblock {\em Bulletin of the {EATCS}}, 105:41--72, 2011.

\bibitem{DBLP:journals/talg/LokshtanovMS18}
D.~Lokshtanov, D.~Marx, and S.~Saurabh.
\newblock Known algorithms on graphs of bounded treewidth are probably optimal.
\newblock {\em {ACM} Trans. Algorithms}, 14(2):13:1--13:30, 2018.

\bibitem{luczak2004note}
T.~{\L}uczak and J.~Ne{\v{s}}et{\v{r}}il.
\newblock Note on projective graphs.
\newblock {\em Journal of Graph Theory}, 47(2):81--86, 2004.

\bibitem{McKenzie1971}
R.~McKenzie.
\newblock Cardinal multiplication of structures with a reflexive relation.
\newblock {\em Fundamenta Mathematicae}, 70(1):59--101, 1971.

\bibitem{miller_1968}
D.~J. Miller.
\newblock The categorical product of graphs.
\newblock {\em Canadian Journal of Mathematics}, 20:1511–1521, 1968.

\bibitem{DBLP:conf/soda/OkrasaR20}
K.~Okrasa and P.~Rz{\k{a}}{\.{z}}ewski.
\newblock Fine-grained complexity of graph homomorphism problem for
  bounded-treewidth graphs.
\newblock In S.~Chawla, editor, {\em Proceedings of the 2020 {ACM-SIAM}
  Symposium on Discrete Algorithms, {SODA} 2020, Salt Lake City, UT, USA,
  January 5-8, 2020}, pages 1578--1590. {SIAM}, 2020.

\bibitem{10.1007/978-3-642-22993-0_47}
M.~Pilipczuk.
\newblock Problems parameterized by treewidth tractable in single exponential
  time: {A} logical approach.
\newblock In {\em {MFCS} 2011}, volume 6907, pages 520--531. Springer, 2011.

\bibitem{DBLP:journals/jct/RobertsonS84}
N.~Robertson and P.~D. Seymour.
\newblock Graph minors. {III. Planar} tree-width.
\newblock {\em J. Comb. Theory, Ser. {B}}, 36(1):49--64, 1984.

\bibitem{DBLP:journals/ipl/Rzazewski14}
P.~Rz{\k{a}}{\.{z}}ewski.
\newblock Exact algorithm for graph homomorphism and locally injective graph
  homomorphism.
\newblock {\em Inf. Process. Lett.}, 114(7):387--391, 2014.

\bibitem{sauer2001hedetniemi}
N.~Sauer.
\newblock Hedetniemi's conjecture—a survey.
\newblock {\em Discrete Mathematics}, 229(1-3):261--292, 2001.

\bibitem{shitov2019counterexamples}
Y.~Shitov.
\newblock Counterexamples to hedetniemi's conjecture.
\newblock {\em Annals of Mathematics}, 190(2):663--667, 2019.

\bibitem{Siggers}
M.~H. Siggers.
\newblock A new proof of the {$H$}-coloring dichotomy.
\newblock {\em SIAM Journal on Discrete Mathematics}, 23(4):2204--2210, 2010.

\bibitem{smith1971primitive}
D.~H. Smith.
\newblock Primitive and imprimitive graphs.
\newblock {\em The Quarterly Journal of Mathematics}, 22(4):551--557, 1971.

\bibitem{tardif2008hedetniemi}
C.~Tardif.
\newblock Hedetniemi’s conjecture, 40 years later.
\newblock {\em Graph Theory Notes NY}, 54(46-57):2, 2008.

\bibitem{DBLP:conf/esa/RooijBR09}
J.~M.~M. van Rooij, H.~L. Bodlaender, and P.~Rossmanith.
\newblock Dynamic programming on tree decompositions using generalised fast
  subset convolution.
\newblock In A.~Fiat and P.~Sanders, editors, {\em Algorithms - {ESA} 2009,
  17th Annual European Symposium, Copenhagen, Denmark, September 7-9, 2009.
  Proceedings}, volume 5757 of {\em Lecture Notes in Computer Science}, pages
  566--577. Springer, 2009.

\bibitem{DBLP:journals/mst/Wahlstrom11}
M.~Wahlstr{\"{o}}m.
\newblock New plain-exponential time classes for graph homomorphism.
\newblock {\em Theory Comput. Syst.}, 49(2):273--282, 2011.

\bibitem{weichsel1962kronecker}
P.~M. Weichsel.
\newblock The {K}ronecker product of graphs.
\newblock {\em Proceedings of the American mathematical society}, 13(1):47--52,
  1962.

\end{thebibliography}

\end{document}